\documentclass{amsart}
\usepackage{times}
\usepackage{amsthm}
\usepackage{amssymb}
\usepackage[pdftex]{graphicx}
\usepackage[all]{xy}
\usepackage[mathscr]{eucal}
\usepackage{amsmath,latexsym,oldlfont}
\usepackage{mathdots}
\usepackage{pifont}
\usepackage{stmaryrd}
\usepackage{textcomp}
\usepackage{pifont}
\usepackage{multirow}
\usepackage{tikz}
\usetikzlibrary{arrows.meta}
\usepackage{float}
\usepackage{bbm}
\usepackage[colorlinks,linkcolor=blue,citecolor=red,linktocpage=true]{hyperref}
\usepackage{mathtools}
\usepackage{pgfkeys}
\usetikzlibrary{positioning}
\numberwithin{equation}{section} 
\newtheorem{lem}{Lemma}
\newtheorem{thm}[lem]{Theorem}
\newtheorem{cor}[lem]{Corollary}
\newtheorem{prop}[lem]{Proposition}
\newtheorem{example}[lem]{Example}
\theoremstyle{definition}
\newtheorem{defn}{Definition}
\theoremstyle{remark}
\newtheorem{rem}{Remark}

\newsavebox{\spacebox}
\begin{lrbox}{\spacebox}
    \verb*! !
\end{lrbox}

\title[ LDPC Fractus Codes: Sparse Codes with Recursive Structure ] { LDPC Fractus Codes: Sparse Codes with Recursive Structure.}
\author{
{\small Jes\'us Carrillo--Pacheco}\\
{\tiny
Academy of Mathematics\\
Universidad Aut\'onoma de la Ciudad de M\'exico\\
09390 M\'exico, Ciudad de M\'exico\\
\texttt{jesus.carrillo@uacm.edu.mx}}}
\begin{document}

\maketitle
\pagestyle{plain}

\begin{abstract}
We introduce a new family of recursively constructed sparse matrices, termed Fractus matrices, and investigate their use in constructing low-density parity-check (LDPC) codes. Generated through a self-similar recursive process, these matrices yield regular sparse parity-check matrices while preserving key structural properties across successive iterations. This recursive structure enables an efficient encoding algorithm with computational complexity that is nearly linear in the block length. Decoding is performed using standard iterative message-passing algorithms, thereby retaining the low-complexity decoding characteristic of LDPC codes.
The proposed construction produces Tanner graphs with girth six and guarantees a minimum Hamming distance of at least $\ell+1$. We establish several algebraic properties of Fractus matrices, including sparsity, regularity, recursive decomposition, and symmetry under the flip-transpose operation. In addition, we show that the family of Fractus matrices admits a natural lattice structure and that the associated LDPC codes inherit corresponding lattice-theoretic properties. These results establish a connection between order theory and coding theory.
Overall, the proposed framework integrates recursive matrix constructions, efficient encoding, graph-theoretic analysis, and lattice theory into a unified algebraic approach to the design and analysis of scalable LDPC codes.
\end{abstract}

\section{Introduction}
Parity-check codes constitute a broad class of error-correcting codes whose description in terms of parity-check matrices is particularly useful for determining code parameters and developing efficient encoding and decoding algorithms. A prominent class within this family is that of low-density parity-check (LDPC) codes.

LDPC codes, together with their associated iterative decoding algorithms, were introduced by Gallager in \cite{3.01,3}. Although they initially received limited attention because of the computational limitations of the time, LDPC codes experienced a renewed interest in the 1990s, when advances in computational power made Gallager's approach practically feasible. Since then, LDPC codes have become an important class of error-correcting codes because of their sparse representation, efficient iterative decoding, and error-correction performance approaching the Shannon limit.

An LDPC code is defined by a sparse parity-check matrix, that is, a matrix whose entries are predominantly zero. More precisely, an $(N,\ell,k)$-regular LDPC code is a linear block code of length $N$ whose parity-check matrix $H$ has exactly $k$ ones in each row and exactly $\ell$ ones in each column. The sparsity of $H$ is one of the main features that makes iterative decoding computationally attractive. In contrast to classical linear block codes, the parity-check matrix of an LDPC code is not necessarily presented in a systematic or diagonal form. Nevertheless, encoding can be performed by solving the parity-check equations represented by $H$ and expressing the parity symbols as linear combinations of the information symbols.

The performance of an LDPC code is determined not only by its length, dimension, and rate, but also by the combinatorial structure of its parity-check matrix. In particular, the Tanner graph associated with $H$ plays a fundamental role in the analysis of iterative decoding. The Tanner graph is a bipartite graph whose variable nodes correspond to the code symbols and whose check nodes correspond to the parity-check equations. This representation provides a natural framework for studying structural properties such as cycles, girth, stopping sets, and degree distributions, as well as for implementing message-passing decoding algorithms such as belief propagation and the sum-product algorithm.

The girth of a Tanner graph is the length of its shortest cycle. Since Tanner graphs are bipartite, their girth is necessarily even. Iterative decoding is optimal on cycle-free Tanner graphs, and hence Tanner graphs with large girth are generally desirable \cite{5.01}. In particular, avoiding cycles of length four is an important design criterion, so that a girth of at least six is typically required in practical LDPC constructions \cite{1.1}. Short cycles introduce statistical dependencies among the messages exchanged during iterative decoding and can therefore degrade decoder performance. Over the binary erasure channel, decoding failures are additionally characterized by stopping sets, which are combinatorial structures of the Tanner graph.

It is important to emphasize that these graph-theoretic properties depend on the particular parity-check matrix used to represent the code. Indeed, a linear code may admit several distinct parity-check matrices,
${\EuScript C}=\ker H_1=\ker H_2=\cdots,$
and consequently several different Tanner graphs, denoted by $T_{H_1},T_{H_2},\ldots$. For example, adding redundant rows, that is, linear combinations of existing rows, to a parity-check matrix produces another parity-check matrix defining the same linear code, but its Tanner graph contains additional check nodes. Thus, properties such as the number of cycles, the girth, the degree distribution, and the behavior of iterative decoding algorithms are properties of a particular Tanner-graph representation rather than intrinsic properties of the abstract linear code. In particular, two parity-check matrices representing the same code may have Tanner graphs with different girths.

Another important issue in the practical implementation of LDPC codes is encoding. Although decoding has traditionally received greater attention, efficient encoding is essential for achieving good overall system performance \cite{4}. If $G$ is a generator matrix for an LDPC code, a codeword corresponding to a message $m$ is obtained as $c=mG$. However, the sparsity of the parity-check matrix $H$ does not generally imply sparsity of the generator matrix $G$. Consequently, direct encoding using $G$ may require considerable memory and computational resources, particularly for codes of large length. This motivates the search for structured parity-check matrices that allow efficient encoding while retaining the advantages of sparse iterative decoding.

Structured constructions of LDPC codes have therefore attracted considerable attention. Many classical constructions rely on algebraic or combinatorial techniques, while others are based on random procedures. Although random constructions can produce codes with desirable asymptotic properties, their lack of explicit structure may complicate the storage and manipulation of large parity-check matrices, as well as the analysis of their associated Tanner graphs. In contrast, algebraically structured constructions can provide compact descriptions and facilitate the study of parameters such as girth, minimum distance, encoding complexity, and decoding performance.

Lattice-theoretic methods provide another algebraic perspective for studying structured error-correcting codes. The ordered structure of lattices offers a natural framework for describing hierarchical relationships, algebraic symmetries, and decompositions that may arise in structured code constructions. In this context, lattice-theoretic concepts can complement graph-theoretic methods in the analysis of LDPC codes and provide additional algebraic tools for studying their underlying structures \cite{1}.

Fractal matrices constitute a class of structured matrices characterized by recursive or self-similar patterns generated through iterative construction rules. Their hierarchical organization makes it possible to construct large matrices from relatively simple building blocks while preserving structural properties such as sparsity and regularity. These features suggest a natural connection with LDPC codes, whose performance depends strongly on the combinatorial structure of sparse parity-check matrices. In particular, recursive matrix constructions can provide systematic ways of controlling row and column weights and, consequently, the degree distributions of the associated Tanner graphs. They may also provide compact descriptions of large parity-check matrices and facilitate the analysis of their graph-theoretic properties, including cycles and girth \cite{6}.

The matrices considered in this paper, which we call \emph{Fractus matrices}, form a particular class of fractal matrices. Their recursive and self-similar structure allows large matrices to be constructed from smaller ones in a systematic manner. This structure is particularly suitable for the construction of regular LDPC codes. In addition, it leads to recursive encoding procedures whose computational complexity is nearly linear in the code length. Decoding can be performed using standard iterative message-passing algorithms, whose complexity is proportional to the number of edges in the corresponding Tanner graph. Since the Tanner graphs considered here are sparse, this number grows linearly with the code length, resulting in efficient decoding for increasingly large codes.

The main purpose of this paper is to investigate the algebraic and combinatorial properties of Fractus matrices and their applications to LDPC codes. We show that the Tanner graphs arising from the proposed constructions have girth six and that the minimum Hamming distance of the resulting codes grows linearly with the code length. The corresponding codes also have competitive rates. Furthermore, we investigate the lattice structures induced by these matrices and establish algebraic properties that reveal a connection between Fractus matrices, lattice theory, and LDPC codes. These results provide a structured framework for constructing and analyzing families of LDPC codes while preserving explicit algebraic and recursive properties.

In classical LDPC coding, the encoding process is usually formulated as a linear algebra problem associated with the parity-check matrix. Given a parity-check matrix $H$, the information symbols are selected and the remaining symbols are determined by solving a system of linear equations. Although this approach is general, the encoding procedure does not necessarily exploit the particular structure of the parity-check matrix. In particular, for a large and sparse LDPC matrix, a direct implementation may require additional algebraic operations or a suitable decomposition of $H$.

The encoding procedure proposed for the Fractus codes is fundamentally different in this respect. The parity-check matrices of the Fractus codes possess a recursive and highly structured form, inherited from the recursive construction of the Fractus matrices. The proposed encoder explicitly exploits this structure rather than treating the parity-check matrix as an arbitrary sparse matrix. Through the successive fragmentations of the Fractus matrix, the encoding problem is reduced recursively to smaller systems, allowing the parity symbols to be determined from the information symbols by explicit linear transformations.

More precisely, the Fractus encoder is a systematic encoder of the form
$$S\longmapsto \bigl(P_2(S),P_1(S),S\bigr),$$
where $S$ contains the information symbols and $P_1(S)$ and $P_2(S)$ are obtained recursively from the Fractus structure. Thus, the information symbols appear explicitly in the codeword, while the parity symbols are computed from them without solving a new general linear system for each codeword.
The main advantage of this construction is therefore not merely that the Fractus codes admit a systematic encoding, since systematic encoding is standard for linear codes. Rather, the distinctive feature is that the systematic encoder is derived from the intrinsic recursive structure of the Fractus matrices. Consequently, the same structure that defines the parity-check matrices also provides a natural mechanism for encoding. This establishes a direct connection between the combinatorial structure of the Fractus matrices and the computational structure of their associated LDPC codes.\\
This article is organized as follows:The paper is organized as follows. Section 3 introduces the definition of Fractus codes and establishes their main properties. In particular, we study the algebraic and combinatorial properties arising from the underlying Fractus matrices and discuss their relevance to the construction of linear codes. Section 4 introduces the fragmentation of codewords and develops the corresponding structure induced by the fragmentation of Fractus matrices. Section 5 studies the lattices associated with Fractus matrices and codes, establishing the order-theoretic relationships induced by the corresponding operations and their interaction with the map $\rho$. Finally, Section 6 presents the main conclusions and outlines some possible directions for further research.
\section{Notation and preliminaries}
 Throughout this paper, we use the notation $C_n^m$ for the binomial coefficient $\binom{m}{n}$. We will frequently use Pascal's identity, which, in our notation, takes the form
\begin{equation}\label{HHYTR}
C_n^m=C_n^{m-1}+C_{n-1}^{m-1}.
\end{equation}
  Let $A=(a_{ij})$ be an $m\times n$ matrix. If every entry $a_{ij}\in \{0,1\}$, then $A$ is called a  $(0, 1)$-matrix. 
   We say that it is \textit{ sparse} if it has many zeros and few ones. It is  $(k, \ell)$-regular if it  is a $(0,1)$-matrix which has exactly $k$ ones in each row and exactly $\ell$ ones in each column; otherwise, the matrix is irregular, see \cite{0.1}, \cite{4.1}.
\begin{defn}\label{fed789865r}
 If $A=(a_{ij})$ be an ${m\times n}$-matrix, we define $A^F=(a_{n-j+1,m-i+1})$  an  ${n\times m}$-matrix and we call it the \textit{ flip-transpose matrix}. 
  \end{defn}
  Let $J_n$ be the $n\times n$ matrix with ones along its skew-diagonal and zeroes elsewhere, that is, $J_n = [e_n | e_{n-1} | \ldots | e_1]$. More over clearly 
\begin{equation}\label{rules000010101}
J_n^2=I_n\; and \; J_n^T=J_n^F=J_n.\\
\end{equation}
The following properties are easy to verify

\begin{equation}\label{rules00001}
\begin{aligned}
\text{(1)}\quad & A_{mn}^F=J_nA_{mn}^TJ_m,\\
\text{(2)}\quad & (A^F)^F=A,\\
\text{(3)}\quad & (A^T)^F=(A^F)^T,\\
\text{(4)}\quad & (A+B)^F=A^F+B^F,\\
\text{(5)}\quad & (AB)^F=B^FA^F.
\end{aligned}
\end{equation}
\begin{lem} \label{FTR43eDD}
Let $A$ be an $m\times n$ matrix. Then the rank of its flip-transpose satisfies
$\operatorname{rank}(A^F)=\operatorname{rank}(A).$
\end{lem}  
\begin{proof}
Considere la propiedad $A^F=J_nA^TJ_m$ then $$\operatorname{rank}(A^F)=\operatorname{rank}(A^T)=\operatorname{rank}(A).$$
\end{proof}  
  A \textit{block lower triangular matrix} (or block-lower triangular matrix) is a square matrix partitioned into submatrices (blocks) such that all blocks above the main block diagonal are zero matrices.  We write $A\unlhd B$ if $A$ is a submatrix of $B$.
 We will give some introductory definitions and background on linear codes. A good reference for additional information is \cite{0.1.0.2},  \cite{3.5.2} and \cite{5.1}.
A Galois field ${\mathbb F}=GF(q)$ is a finite field of cardinality $q=p^m$ with $p$ is a prime.
A linear $[n, k]$ block code ${\EuScript C}$, is a $k$-dimensional subspace of the $n$-vector space $V$ over ${\mathbb F}$. 
A generator matrix $G$ of an $(n, k)$ code ${\EuScript C}$ is a $k\times n$ matrix whose rows are linearly independent and 
satisfies for all $c\in {\EuScript C}$ we have to $c=uG$
 
\begin{defn}\label{Genertmtrix 0980}
Let  $H\in{\mathbb F}^{m\times n}$ be a sparse $(k, \ell)$-regular, parity-check matrix. The \textit{Low-Density Parity-Check} (LDPC) code defined by $H$ is the linear code
\begin{equation}\label{equ100001}
{\EuScript C}_k^{\ell}=\{ x\in{\mathbb F}^n | Hx^T=0 \} 
\end{equation}
In other words, a vector $x=(x_1,x_2,\ldots,x_n)$ is a codeword if and only if it satisfies all parity-check equations represented by the rows of $H$.
\end{defn}
Throughout this work, we also use the notation
\begin{equation}\label{equ1222221}
{\EuScript C}_k^{\ell}=\ker H,
\end{equation}
where ${\EuScript C}$ is the code defined in \eqref{equ100001}.
If $H$ has rank $R$, then the code has length $n$, dimension $k=n-R$ and rate 
\begin{equation}\label{FRED453}
r=\frac{k}{n}=1-\frac{R}{n}.
\end{equation}
When $H$ has full row rank, i.e., $R=m$, the dimension is $k=n-m$.
The defining feature of an LDPC code is that its parity-check matrix is sparse; that is, each row and each column contains only a small number of nonzero entries relative to the matrix dimensions. This sparsity enables efficient iterative decoding algorithms based on the associated Tanner graph.
The $m$ by $n$ matrix $H=[a_{ij}]$ is called the \textit{ parity-check matrix}.
First of all, for the encoding process, it is necessary to create a generator matrix G from the parity-check matrix H. 
The generator matrix $G$ can be obtained by applying Gaussian elimination and elementary row operations to the parity-check matrix $H$. If $H$ is in systematic form,
$H=(I_{n-k}\mid A)$, where $A$ is an $(n-k)\times k$ matrix, then the corresponding generator matrix is
\begin{equation}\label{equattwo}
G = (-A^T\mid I_k),
\end{equation}
where $I_k$  is the $k\times k$ identity matrix and $A^T$ denotes the transpose of $A$.
\begin{defn}\label{TYTR6y54}
The encoder is called systematic if the $K$ information symbols $S$ appear unchanged as $K$ coordinates of the codeword.
\end{defn}
A lattice ${\mathscr L}$ is a partially ordered set in which every pair $a, b\in {\mathscr L}$ of elements has a unique supremum or join  (least upper bound) $a\vee b$, 
i.e. $a\leq a\vee b$  and  $b\leq a\vee b$ also, if $u$ is any common upper bound of $a$ and $b$, then $a\vee b \leq u$.\\
Analogously to infimum or meet (greatest lower bound) $a\wedge b$ 
i.e. $a\wedge b \leq a$  and  $a\wedge b \leq b$ also, if $u$ is any common lower bound of $a$ and $b$, then $u \leq a\wedge b $.\\
And it satisfies the following properties:\\
 Commutativity: $a\vee b=b\vee a$ and $a\wedge b=b\wedge a$\\
 Associativity: $(a\vee b)\vee c=a\vee (b \vee c)$ and  $(a\wedge b)\wedge c=a\wedge (b\wedge c)$\\
 Idempotence: $a\vee a=a$ and $a\wedge a=a$ \\
 Absorption law: $a\vee (a\wedge b)=a$ and $a\wedge (a\vee b)=a$\\
These properties show that a lattice can be described equivalently either in terms of a partial order or in terms of the two binary operations $\vee$ and $\wedge$. In particular, the partial order can be recovered from the lattice operations by
$$a\leq b
\quad\Longleftrightarrow\quad
a\vee b=b
\quad\Longleftrightarrow\quad
a\wedge b=a.$$
For additional definitions and properties on this topic, see \cite{0.001}.\\
Consequently, the family of Fractus matrices $\mathscr{FM}$ inherits a natural lattice structure through its indices. More precisely, whenever the corresponding operations are defined within the family, the order relation between Fractus matrices is determined by their parameters:
$$H_k^{\ell}\leq H_{k^{\prime}}^{\ell^{\prime}}
\quad\Longleftrightarrow\quad
k\leq k^{\prime}\ \text{and}\ \ell\leq\ell^{\prime}.$$
Thus, the lattice structure organizes the family $\mathscr{FM}$ into a two-parameter hierarchy rather than as an arbitrary collection of binary matrices.

An important feature of Fractus matrices is that their lattice structure is accompanied by a recursive and self-similar structure. The matrices are generated recursively from matrices with smaller parameters, while their dimensions are determined by binomial coefficients. This recursive construction induces a hierarchical pattern that is reflected in both their combinatorial properties and their associated Tanner graphs.

The transpose operation also interacts naturally with this structure. In particular, the Fractus family is closed under the flip-transpose operation, in the sense that
$(H_k^\ell)^F=H_\ell^k,$
where $F$ denotes the corresponding flip-transpose operation. Hence,
${\mathscr FM}^F={\mathscr FM}.$\\
When the parity-check matrix is sparse and its dimensions grow according to the recursive Fractus construction, the resulting codes belong naturally to the family of low-density parity-check (LDPC) codes. The recursive structure of the Fractus matrices is particularly useful in this context. It provides a systematic way to construct parity-check matrices of increasing dimensions while preserving the underlying combinatorial pattern. Moreover, the block structure induced by the recursion can be exploited in the design of encoding and decoding procedures. In particular, suitable fragmentations of a Fractus matrix lead to block decompositions that can be used to obtain efficient encoding algorithms.
Thus, the lattice structure, the recursive construction, the flip-transpose symmetry, and the sparse block structure are not independent features. Together they provide the algebraic and combinatorial framework for the construction of Fractus LDPC codes. This viewpoint allows the parameters $(k,\ell)$ to serve simultaneously as indices of a lattice, as parameters of the underlying matrix family, and as structural parameters of the corresponding error-correcting codes.
\section{ Definition and Properties of Fractus Codes.} 
With the notation $C_{\ell}^m=\frac{m!}{\ell! (m-\ell)!}$  for the binomial coefficient, we have
\begin{defn}\label{definit000001}
Let ${\mathbb N}_0=\{0,1,2,\ldots\}$, and let ${\mathbb M}_{(0,1)}$ denote the set of all $(0,1)$-matrices over a field $\mathbb{F}$. We recursively define the $(k,\ell)$-\textit{Fractus matrices} as follows.
\begin{align*}
 \Psi: {\mathbb N}_0\times {\mathbb N}_0&\longrightarrow {\mathbb M}_{(0,1)}\\
                                                        (k, \ell)&\longmapsto \Psi(k, \ell)
  \end{align*}
 \begin{description}
\item[a)] $\Psi(0,0)= [0]$, \quad $ \Psi(k, 0)=[1]$, \quad $\Psi(0, \ell)=[1]$ is the $1\times 1$-matrix respetively.
\item [b)] $\Psi(k,1)=(1,1,\ldots, 1)$ the $1\times k$-matrix and 
$\Psi(1, \ell)=\left[\begin{smallmatrix}
1\\
\vdots \\
1 
\end{smallmatrix}\right]$ the $\ell \times 1$-matrix
\item[c)] If $k, \ell \geq 2$ then we recursively define 
\begin{equation}\label{JUHG456}
\Psi(k, \ell)=\left[ \begin{array}{c|c} \Psi(k, \ell-1) & 0 \\ \hline I_{C_{\ell-1}^{k+\ell -2}} & \Psi(k-1, \ell) \end{array} \right]
\end{equation}
\end{description}
where  $\Psi(k, \ell)$ is a  matrix of size $C_{\ell-1}^{k+\ell-1}\times C_{\ell}^{k+\ell-1}$, $0$ denote the zero matrix
of size $C_{\ell-2}^{k+\ell-2}\times C_{\ell}^{k+\ell-2}$ and  $I_{C_{\ell-1}^{k+\ell -2}}$ is the identity matrix of size $C_{\ell-1}^{k+\ell -2}\times C_{\ell-1}^{k+\ell -2}$.
\end{defn}
 In this article, we use the notation
\begin{equation}\label{GAT543k}
H_k^{\ell}:=\Psi(k,\ell).
\end{equation}
Using this notation, we can restate the definition \ref{definit000001}.
\begin{defn}\label{def56430001}
For all  $(k, \ell)\in {\mathbb N}_0\times {\mathbb N}_0$ , we define the $(k, \ell)$-\textit{fractus matrices}
\begin{description}
\item[a)] $H_0^0= [0]$, \quad $H_k^0=[1]$, \quad $H_0^{\ell}=[1]$ is the $1\times 1$-matrix respetively.
\item [b)] $H_k^1=(1, \ldots, 1)$ the $1\times k$-matrix and 
$H_1^{\ell}=\left[\begin{smallmatrix}
1\\
\vdots \\
1 
\end{smallmatrix}\right]$ the $\ell \times 1$-matrix
\item[c)] For $k, \ell \geq 2$, $H_k^{\ell}$ is a $(k, \ell)$-regular matrix of size  $C_{\ell-1}^{k+\ell-1}\times C_{\ell}^{k+\ell-1}$
\begin{equation}\label{m4tr1x67532}
H_k^{\ell}=  \left[ \begin{tabular}{c|c}
 $H_k^{\ell-1}$ & $ 0 $ \cr\hline $I_{C_{\ell-1}^{k+\ell -2}}$  & $H_{k-1}^{\ell}$
\end{tabular} \right] ,
\end{equation}
where $H_k^{\ell-1}$ is a $(k, \ell-1)$-regular submatrix of size $C_{\ell-2}^{k+\ell-2}\times C_{\ell-1}^{k+\ell-2}$, $H_{k-1}^{\ell}$ is a $(k-1, \ell)$-regular matrix of size $C_{\ell-1}^{k+\ell-2}\times C_{\ell}^{k+\ell-2}$,
$I_{C_{\ell-1}^{k+\ell -2}}$ denotes the identity matrix of order $C_{\ell-1}^{k+\ell -2}\times C_{\ell-1}^{k+\ell -2}$ and  0 denotes the zero matrix of size  $C_{\ell-2}^{k+\ell-2}\times C_{\ell}^{k+\ell-2}$.
\end{description}
\end{defn}
\begin{example}\label{gssto432}
In these examples, we illustrate the structure of $(k, \ell)$-Fractus matrices. 
\begin{enumerate}
\item
$H_2^2= \left[\begin{tabular}{c|c}
$H_2^1$ & $0$ \cr\hline $I_2$  & $H_1^2$
\end{tabular} \right]$
$= \left[ \begin{tabular}{cc|c}
 1 &1& 0 \cr\hline
 1 & 0 & 1\\
 0 & 1 & 1 \\
\end{tabular} \right]$ 
\newline
\newline
\item
$ H_3^2= \left[\begin{tabular}{c|c}
 $H_3^1$ & $ 0 $ \cr\hline $I_3$  & $H_2^2$
\end{tabular} \right]$
$= \left[\begin{tabular}{ccc|ccc}
1 & 1 & 1 & 0 & 0 & 0  \cr\hline
1 & 0 & 0 & 1 & 1 & 0 \\
0 & 1 & 0 & 1 & 0 & 1 \\
0 & 0 & 1 & 0 & 1 & 1
\end{tabular} \right]$\\
\newline
\newline
\item
$H_2^3= \left[ \begin{tabular}{c|c}
 $H_2^2$ & $ 0 $ \cr\hline $I_3$  & $H_1^3$
\end{tabular} \right] $
$= \left[ \begin{tabular}{ccc|c}
1 & 1 &  0 & 0 \\
1 & 0 & 1 &  0 \\
0 & 1 & 1 & 0  \cr\hline
1 & 0 & 0 & 1 \\
0 & 1 & 0 & 1 \\
0 & 0 & 1 & 1  \\
\end{tabular} \right]$\\
\newline
\newline
\item
$ H_3^3= \left[\begin{tabular}{c|c}
 $H_3^2$ & $ 0 $ \cr\hline $I_6$  & $H_2^3$
\end{tabular} \right]=$
$\left[\begin{tabular}{cccccc|cccc}
1 & 1 & 1 & 0 & 0 & 0 & 0 & 0 & 0& 0 \\
1 & 0 & 0 & 1 & 1 & 0 & 0 & 0 & 0& 0 \\
0 & 1 & 0 & 1 & 0 & 1 & 0 & 0 & 0& 0\\
0 & 0 & 1 & 0 & 1 & 1 & 0 & 0 & 0& 0\cr\hline
1 & 0 & 0 & 0 & 0 & 0 & 1 & 1 & 0& 0\\
0 & 1 & 0 & 0 & 0 & 0 & 1 & 0 & 1& 0\\
0 & 0 & 1 & 0 & 0 & 0 & 0 & 1 & 1& 0\\
0 & 0 & 0 & 1 & 0 & 0 & 1 & 0 & 0 & 1\\
0 & 0 & 0 & 0 & 1 & 0 & 0 & 1 & 0 & 1\\
0 & 0 & 0 & 0 & 0 & 1 & 0 & 0 & 1 & 1\\
\end{tabular} \right]$\\
\newline
\newline
\item
$H_4^2= \left[ \begin{tabular}{c|c}
 $H_4^1$ & $ 0 $ \cr\hline $I_4$  & $H_3^2$
\end{tabular} \right]$
$=\left[ \begin{tabular}{cccc|cccccc}
1 & 1 & 1 & 1 & 0 & 0 & 0 & 0 & 0 & 0\cr\hline
1 &    0     &  0 & 0 & 1 & 1 &  1 & 0 & 0 & 0\\
        0 & 1 & 0 & 0 & 1 & 0 &  0 & 1 & 1 &  0\\
        0 & 0 & 1 & 0 & 0 & 1 &  0 & 1 & 0 &  1\\
0 & 0 & 0 & 1 & 0 & 0 &  1 & 0 & 1 &  1
\end{tabular} \right]$
\newline
\newline\\
\item
$H_3^3= \left[ \begin{tabular}{c|c}
 $H_3^2$ & $ 0 $ \cr\hline $I_6$  & $H_2^3$
\end{tabular} \right]
= \left[ \begin{tabular}{cccccc|cccc}
1 & 1 & 1 & 0 & 0 & 0 & 0 & 0 & 0& 0 \\
1 & 0 & 0 & 1 & 1 & 0 & 0 & 0 & 0& 0 \\
0 & 1 & 0 & 1 & 0 & 1 & 0 & 0 & 0& 0\\
0 & 0 & 1 & 0 & 1 & 1 & 0 & 0 & 0& 0\cr\hline
1 & 0 & 0 & 0 & 0 & 0 & 1 & 1 & 0& 0\\
0 & 1 & 0 & 0 & 0 & 0 & 1 & 0 & 1& 0\\
0 & 0 & 1 & 0 & 0 & 0 & 0 & 1 & 1& 0\\
0 & 0 & 0 & 1 & 0 & 0 & 1 & 0 & 0 & 1\\
0 & 0 & 0 & 0 & 1 & 0 & 0 & 1 & 0 & 1\\
0 & 0 & 0 & 0 & 0 & 1 & 0 & 0 & 1 & 1\\
\end{tabular} \right]$ 
\end{enumerate}
\end{example}
\begin{defn}\label{Akl4327}
Let $k, \ell\geq 2$  be  integers, let $H_k^{\ell}$ be a $(k,\ell)$-fractus matrix,  we defined by $(k, \ell)$-\textit{Fractus Codes}
\begin{enumerate}
\item ${\EuScript C}_k^{\ell}=0$ \quad if $(k, \ell)\in \{ (k, 0), (0, \ell)\}$ \quad where $k, \ell \geq 1$,
\item ${\EuScript C}_k^{\ell}=\{X\in {\mathbb F}^{C_{\ell}^{k+\ell-1}} : H_k^{\ell}X^T=0^T\}$ \quad if $k, \ell \geq 2$.
\end{enumerate}
 \end{defn}
  Thus,\\ 
  ${\EuScript C}_k^{\ell}$ is a linear block code of length $N=C_{\ell}^{k+\ell-1}$ if $k, \ell \geq 2$,\\
  ${\EuScript C}_k^{\ell}$ is a linear block code of length $N=0$  if  $(k, \ell)\in \{ (k, 0), (0, \ell)\}$ where $k, \ell \geq 1$ \\
  and $H_k^{\ell}$ is its $(k, \ell)$-\textit{parity-check fractus matrix}.
\begin{lem}\label{regul44556}
If $k, \ell \geq 2$, then ${\EuScript C}_k^{\ell}$ is $(k, \ell)$-regular code.
\end{lem}
\begin{proof}
The proof is by induction on $k$ and $\ell$.\\
Clearly $H_2^2$ is $(2, 2)$-regular.
Now suppose that for each $2\leq k^{\prime}<k$ and for each $2\leq \ell^{\prime}<\ell$ we have that $H_{k^{\prime}}^2$ is $(k^{\prime}, 2)$-regular and
$H^{\ell^{\prime}}_2$ is $(2, \ell^{\prime})$-regular then clearly
$ H_k^2 = \left[ \begin{tabular}{c|c}
 $H_k^1$ & $ 0 $ \cr\hline $I_k$  & $H_{k-1}^2$
\end{tabular} \right]$
is a $(k, 2)$-regular analogously
$ H_2^{\ell} = \left[ \begin{tabular}{c|c}
 $H_2^{\ell-1}$ & $ 0 $ \cr\hline $I_{\ell-1}$  & $H_1^{\ell}$
\end{tabular} \right]$
is a $(2, \ell)$-regular.\\
The induction hypothesis states that for all $2\leq k^{\prime}<k$ and for all $2\leq \ell^{\prime}<\ell$,
the matrix $H_k^{\ell^{\prime}}$ is $(k, \ell^{\prime})$-regular and   $H_{k^{\prime}}^{\ell}$ is $(k^{\prime}, \ell)$-regular. \\
Then, by definition
\begin{center}
$ H_k^{\ell} = \left[ \begin{tabular}{c|c}
 $H_k^{\ell-1}$ & $ 0 $ \cr\hline $I_{C_{\ell-1}^{k+\ell -2}}$  & $H_{k-1}^{\ell}$
\end{tabular} \right].$
\end{center}
By the induction hypothesis, $H_k^{\ell-1}$ is $(k, \ell-1)$-regular and $H_{k-1}^{\ell}$ is $(k-1, \ell)$-regular. Therefore, it follows immediately that 
$H_k^{\ell}$ is $(k, \ell)$-regular
\end{proof}
\begin{lem}\label{sp4r539999}
For all $k,\ell\geq 2$, the Fractus code ${\EuScript C}_k^{\ell}$ is an $(N,\ell,k)$-LDPC code.
\end{lem}
\begin{proof}
By construction (see Definition  \ref{definit000001}),  $H_k^{\ell}$ is a $C_{\ell-1}^{k+\ell-1}\times C_{\ell}^{k+\ell-1}$-matrix  and is 
$(k, \ell)$-regular. Hence the density of ones in the matrix is 
$$\frac{kC_{\ell-1}^{k+\ell-1}}{C_{\ell-1}^{k+\ell-1}\times C_{\ell}^{k+\ell-1}}=\frac{\ell C_{\ell}^{k+\ell-1}}{C_{\ell-1}^{k+\ell-1}\times C_{\ell}^{k+\ell-1}}\leq 1$$ then
$$\lim_{k \to \infty} \frac{k}{C_{\ell}^{k+\ell-1}}=\lim_{\ell \to \infty}\frac{\ell }{C_{\ell-1}^{k+\ell-1}}=0$$ 
Thus, as shown in the proof of Proposition \ref{regul44556}, the matrix $H_k^{\ell}$ is sparse and $(k,\ell)$-regular. Consequently, ${\EuScript C}_k^{\ell}$ is an $(N,\ell,k)$-LDPC code.
\end{proof}
\begin{defn}\label{obsyyt65}
Let $k,\ell\geq 2$ be integers, let $0\leq r\leq \min\{k-1,\ell-1\},$ and let $H_k^{\ell}$ be a $(k,\ell)$-fractus matrix. For $1\leq r\leq \min\{k-1,\ell-1\},$ we recursively define  
the \textit{$r$-fragmentation}  of $H_k^\ell$,  denoted by $H_k^{\ell}(r)$, 
\begin{itemize}
\item $H_k^{\ell}(0)=H_k^{\ell}$,
\item
$H_k^{\ell}(r)=
\left[
\begin{array}{c|c}
H_k^{\ell-1}(r-1) & 0 \\ \hline
I_{C_{\ell-1}^{k+\ell-2}} & H_{k-1}^{\ell}(r-1)
\end{array}
\right].$
\end{itemize}
\end{defn}
\begin{example}\label{GTR543hu}
In this example, we illustrate the first four fragmentations of a $(k,\ell)$-Fractus matrices.
\begin{equation*}\label{r87}
H_{k}^{\ell}(0)=H_{k}^{\ell}
\end{equation*}
\begin{equation*}\label{redre67987}
H_{k}^{\ell}(1)= \left[ \begin{tabular}{c|c}
 $H_k^{\ell-1}(0)$ & $ 0 $ \cr\hline $I_{C_{\ell-1}^{k+\ell-2}}$  & $H_{k-1}^{\ell}(0)$
\end{tabular} \right]
= \left[ \begin{tabular}{c|c}
 $H_k^{\ell-1}$ & $ 0 $ \cr\hline $I_{C_{\ell-1}^{k+\ell-2}}$  & $H_{k-1}^{\ell}$
\end{tabular} \right]
\end{equation*}
\begin{equation*}\label{redre6754}
H_{k}^{\ell}(2)= \left[ \begin{tabular}{c|c}
 $H_k^{\ell-1}(1)$ & $ 0 $ \cr\hline $I_{C_{\ell-1}^{k+\ell-2}}$  & $H_{k-1}^{\ell}(1)$
\end{tabular} \right]= 
\left[ \begin{tabular}{c|c|c|c}
 $H_k^{\ell-2}$ & $0$ & $0$ & $0$ 
\cr\hline $I_{C_{\ell-2}^{k+\ell-3}}$ & $H_{k-1}^{\ell-1}$ & $0$ & $0$
 \cr\hline $I_{C_{\ell-2}^{k+\ell-3}}$ & $0$ & $H_{k-1}^{\ell-1}$ & $0$ 
\cr\hline $0$ & $I_{C_{\ell-1}^{k+\ell-3}}$ & $I_{C_{\ell-1}^{k+\ell-3}}$ & $H_{k-2}^{\ell}$
\end{tabular} \right]
\end{equation*}
\begin{equation*}\label{redreyyuy}
H_{k}^{\ell}(3)= \left[ \begin{tabular}{c|c|c|c|c|c|c|c}
 $H_k^{\ell-3}$ & $0$ & $0$ & $0$ & $0$ & $0$ & $0$ & $0$
\cr\hline $I_{C_{\ell-3}^{k+\ell-4}}$ & $H_{k-1}^{\ell-2}$ & $0$ & $0$ & $0$ & $0$ & $0$ & $0$
 \cr\hline $I_{C_{\ell-3}^{k+\ell-4}}$ & $0$ & $H_{k-1}^{\ell-2}$ & $0$ & $0$ & $0$ & $0$ & $0$
\cr\hline $0$ & $I_{C_{\ell-2}^{k+\ell-4}}$ & $I_{C_{\ell-2}^{k+\ell-4}}$ & $H_{k-2}^{\ell-1}$ & $0$ & $0$ & $0$ & $0$
\cr\hline $I_{C_{\ell-3}^{k+\ell-4}}$ & $0$ & $0$ & $0$ & $H_{k-1}^{\ell-2}$ & $0$ & $0$ & $0$
\cr\hline $0$ & $I_{C_{\ell-2}^{k+\ell-4}}$ & $0$ & $0$ & $I_{C_{\ell-2}^{k+\ell-4}}$ & $H_{k-2}^{\ell-1}$ & $0$ & $0$
\cr\hline $0$ & $0$ & $I_{C_{\ell-2}^{k+\ell-4}}$ & $0$ & $I_{C_{\ell-2}^{k+\ell-4}}$ & $0$ & $H_{k-2}^{\ell-1}$ & $0$
\cr\hline $0$ & $0$ & $0$ & $I_{C_{\ell-1}^{k+\ell-4}}$ & $0$ & $I_{C_{\ell-1}^{k+\ell-4}}$ & $I_{C_{\ell-1}^{k+\ell-4}}$ & $H_{k-3}^{\ell}$
\end{tabular} \right]
\end{equation*}
\end{example}
\bigskip
\subsection{Systematic linear encoder induced by the Fractus structure.}\label{resoleq2312} 
Let $k, \ell\geq 2$  
 the objective of this section first is to solve the matrix system.
\begin{equation}\label{Carm765101}
H_k^{\ell}X^T=Y^T
\end{equation}
 for a given vector $Y\in {\mathbb F}^{C_{\ell-1}^{k+\ell-1}}$. We denote by
\begin{equation}\label{Carm765}   
{\mathcal S}_{(k, \ell, Y)}=\{X\in {\mathbb F}^{C_{\ell}^{k+\ell-1}}:  H_k^{\ell} X^T=Y^T\} 
\end{equation}
the solution set. 
By item (2) of first fragmentation see Example \ref{tygffd3u76}, we can write the system $H_k^{\ell}X^T=Y^T$ in block-matrix form
\begin{equation}\label{rrted23145}
  \left[ \begin{tabular}{c|c}
 $H_k^{\ell - 1}$ & $ 0 $ \cr\hline $I_{C_{\ell -1}^{k+\ell -2}}$  & $H_{k-1}^{\ell}$
\end{tabular} \right] \begin{bmatrix} 
 X_1^T \\
X_2^T 
\end{bmatrix}
=
\begin{bmatrix} 
 Y_1^T \\
Y_2^T
\end{bmatrix}
\end{equation}
We can always write 
\begin{equation}\label{ABmatrix323}
H_{k-1}^{\ell}=B\sqcup A
\end{equation} 
  where $B$ be a submatrix of order $C^{k+\ell-2}_{\ell-1}\times C^{k+\ell-2}_{\ell-2}$ and  $A$  a submatrix of order  $C^{k+\ell-2}_{\ell-1}\times (C^{k+\ell-2}_{\ell}-C^{k+\ell-2}_{\ell-2})$
  of $H_{k-1}^{\ell}$, recall that $H_k^{\ell}$ is a matrix of size $C^{k+\ell-1}_{\ell-1}\times C^{k+\ell-1}_{\ell}$ where the symbol $\sqcup$ denotes the concatenation of the submatrices $B$ and $A$ to form $H_k^{\ell}$. Then we have
  \begin{equation}\label{Yatr543}
  H_k^{\ell}=\left[ \begin{tabular}{c|c|c}
                   $H_k^{\ell-1}$ & $0$ & $0$ 
\cr\hline $ I_{C^{k+\ell-2}_{\ell-1}}$ & $B$ &  $A$ 
 \end{tabular} \right]
 \end{equation}
\begin{thm}\label{DfDfFr43765}
Let $k,\ell\geq 2$ be integers 
then ${\mathcal S}_{(k,\ell,Y)}$ is the set of all $X=(P_2,P_1,S)$ such that
\begin{enumerate}
\item If the vector $Y=(Y_1, Y_2)\neq 0$. Then  
\begin{equation}\label{soluciones7777}
\begin{aligned}
P_2^T&=Y_2^T-BP_1^T-AS^T \\
P_1^T&=\phi^{-1}\left[Y_1^T-H_k^{\ell-1}Y_2^T+(H_k^{\ell-1}A)S^T\right] \\
            S&\in{\mathbb F}^{C^{k+\ell-2}_{\ell}-C^{k+\ell-2}_{\ell-2}}.
\end{aligned}
\end{equation}
\item If the vector $Y=0$. Then 
\begin{equation}\label{soluciones8888}
\begin{aligned}
P_2^T&=-BP_1^T-AS^T \\
P_1^T&=\phi^{-1}(H_k^{\ell-1}A)S^T \\
            S&\in{\mathbb F}^{C^{k+\ell-2}_{\ell}-C^{k+\ell-2}_{\ell-2}}.
\end{aligned}
\end{equation}
\end{enumerate}
\end{thm}
\begin{proof}
We consider a solution of  type $X=(P_2, P_1, S)$, with 
 $P_2\in{\mathbb F}^{C^{k+\ell-2}_{\ell-1}}$, $P_1\in{\mathbb F}^{C^{k+\ell-2}_{\ell-2}}$,  $S\in{\mathbb F}^{C^{k+\ell-2}_{\ell}-C^{k+\ell-2}_{\ell-2}}$ and  $Y^T=\begin{bmatrix} 
 Y_1^T \\
Y_2^T
\end{bmatrix}$
 be a known column vector of size  $C^{k+\ell-1}_{\ell-1}$ with $Y_1$ of size $1 \times C^{k+\ell-2}_{\ell-2}$ and $Y_2$ of size $1 \times C^{k+\ell-2}_{\ell-1}$. Then of \ref{Yatr543} we have
{\tiny \begin{multline*}
  \begin{bmatrix} 
0\;\; &\;\;  I_{C^{k+\ell-2}_{\ell-1}}  \\
I_{C^{k+\ell-2}_{\ell-2}} & -H_k^{\ell-1}  
\end{bmatrix}
\begin{bmatrix}
H_k^{\ell-1}\;\; & 0\;\; & 0 \;\; \\
 I_{C^{k+\ell-2}_{\ell-1}}\;\;& B\;\; &  A \;\;
\end{bmatrix}
\begin{bmatrix} 
 P_2^T \\
P_1^T \\
S^T
\end{bmatrix}
$=$
  \begin{bmatrix} 
0\;\; &\;\;  I_{C^{k+\ell-2}_{\ell-1}}  \\
I_{C^{k+\ell-2}_{\ell-2}} & -H_k^{\ell-1}  
\end{bmatrix}
\begin{bmatrix} 
 Y_1^T \\
Y_2^T
\end{bmatrix}
\end{multline*}}
{\tiny \begin{multline}
\begin{bmatrix}\label{Fiell5654}
 I_{C^{k+\ell-2}_{\ell-1}}\;\;& B\;\; &  A \;\;\\
0\;\; & -H_k^{\ell-1}B\;\; & -H_k^{\ell-1} A \;\; 
\end{bmatrix}
\begin{bmatrix} 
 P_2^T \\
P_1^T \\
S^T
\end{bmatrix}
$=$
\begin{bmatrix} 
 Y_2^T \\
Y_1^T-H_k^{\ell-1}Y_2^T
\end{bmatrix}
\end{multline}}
gives us $$P_2^T+BP_1^T+AS^T=Y_2^T$$ and  $$-(H_k^{\ell-1}B)P_1^T-(H_k^{\ell-1}A)S^T =Y_1^T-H_k^{\ell-1}Y_2$$ 
Let 
\begin{equation}\label{invM4trx1010}
\phi:=-H_k^{\ell-1}B
\end{equation}\label{solut656543}
 the matrix of order $C^{k+\ell-2}_{\ell-2}\times C^{k+\ell-2}_{\ell-2}$ and suppose  that it is invertible, if it is not invertible, we can simply perform further column permutations to remove this singularity, see section II of  \cite{ 4}.  So now we have,  for a vector $S$ whose values were chosen:
\begin{align*}
P_1^T&=\phi^{-1}\left[Y_1^T-H_k^{\ell-1}Y_2^T+(H_k^{\ell-1}A)S^T\right] \\
P_2^T&=Y_2^T-BP_1^T-AS^T \\
            S&\in{\mathbb F}^{C^{k+\ell-2}_{\ell}-C^{k+\ell-2}_{\ell-2}}.
\end{align*}
then we have $(P_2^T, P_1^T, S)\in {\mathcal S}_{(k, \ell, Y)}$.\\ 
If, in equation  \ref{Carm765101}, $Y=0$ then 
\begin{equation}\label{casez340}
\begin{aligned} 
P_1^T&=\phi^{-1}(H_k^{\ell-1}A)S^T\\
P_2^T&=-BP_1^T-AS^T\\
         S&\in{\mathbb F}^{C^{k+\ell-2}_{\ell}-C^{k+\ell-2}_{\ell-2}}
\end{aligned}
\end{equation}
and  $X=(P_2, P_1, S)\in {\mathcal S}_{(k, \ell, \overline{0})}$ where
${\mathcal S}_{(k, \ell, 0)}=\{X:  H_k^{\ell}X^T=0^T\}$ denote  the solution space of the homogeneous system $H_k^{\ell} X^T=0^T$
\end{proof}
The following definition is based on theorem \ref{DfDfFr43765}.
\begin{defn}\label{HYU654}
Let $k,\ell \geq 2$ be integers, and let $A$ and $B$ be matrices such that
$$
H_{k-1}^{\ell}=B\sqcup A,
$$
with $\phi=-H_k^{\ell-1}B$ an invertible matrix.

For a given $Y=(Y_1,Y_2)\in {\mathbb F}^{C_{\ell-1}^{k+\ell-1}}$, we define the function
\begin{equation}\label{Uuuyt65X8}
\begin{aligned}
\Psi_{k, \ell,Y}: {\mathbb F}^{C^{k+\ell-2}_{\ell}-C^{k+\ell-2}_{\ell-2}}&\longrightarrow {\mathcal S}_{(k, \ell, Y)}\\
                                                                                        S&\longmapsto (P_2(S), P_1(S), S) 
\end{aligned}
\end{equation}
where
\begin{equation}\label{Vvvyt65X8}
\begin{aligned}
P_2^T(S)&=Y_2^T-BP_1^T-AS^T \\
P_1^T(S)&=\phi^{-1}\left[Y_1^T-H_k^{\ell-1}Y_2^T+(H_k^{\ell-1}A)S^T\right] 
\end{aligned}
\end{equation}
\end{defn}
\begin{prop}\label{DfDfFr43}
Let $k,\ell\geq 2$ be integers, and let $\Psi_{k,\ell,Y}$ be the map defined in \ref{Uuuyt65X8}. Then:
\begin{enumerate}
\item $\Psi_{k,\ell,Y}$ is an injective function.
\item If $\operatorname{rank}(H_k^{\ell})=C_{\ell-1}^{k+\ell-1}$, then 
$$\Psi_{k, \ell,\overline{0}}: {\mathbb F}^{C^{k+\ell-2}_{\ell}-C^{k+\ell-2}_{\ell-2}}\longrightarrow {\EuScript C}_k^{\ell}$$ 
is an isomorphism.
\end{enumerate}
\end{prop}
\begin{proof}
For item (1), by Theorem \ref{DfDfFr43765}, every solution vector is of the form
$X=(P_2(S),P_1(S),S),$
where
$S\in{\mathbb F}^{C_{\ell}^{k+\ell-2}-C_{\ell-2}^{k+\ell-2}}.$

Suppose that
$\Psi_{k,\ell,Y}(S_1)
=(P_2(S_1),P_1(S_1),S_1)
=(P_2(S_2),P_1(S_2),S_2)
=\Psi_{k,\ell,Y}(S_2).$
Equality of the third coordinates immediately implies that
$S_1=S_2.$
Therefore, $\Psi_{k,\ell,Y}$ is injective.
 For item (2), suppose that $Y=\overline{0}$. Then
\begin{align*}
P_2^T(S)&=-BP_1^T-AS^T,\\
P_1^T(S)&=\phi^{-1}(H_k^{\ell-1}A)S^T.
\end{align*}
Substituting $P_1(S)$ into the expression for $P_2(S)$, we obtain
\begin{equation}\label{GHTR87hg}
\begin{aligned}
P_2^T(S)
&=-BP_1^T-AS^T\\
&=-B\phi^{-1}(H_k^{\ell-1}A)S^T-AS^T\\
&=-\bigl[B\phi^{-1}(H_k^{\ell-1}A)+A\bigr]S^T.
\end{aligned}
\end{equation}
Clearly, by item (1), $\Psi_{k,\ell,\overline{0}}$ is a monomorphism.
By the  definition \ref{Carm765}, we have ${\EuScript C}_k^{\ell}\cong {\mathcal S}_{(k, \ell, \overline{0})}$
\begin{align*}
\dim {\EuScript C}_k^{\ell}
&=C_{\ell}^{k+\ell-1}-C_{\ell-1}^{k+\ell-1}\\
&=(C_{\ell}^{k+\ell-2}+C_{\ell-1}^{k+\ell-2})-(C_{\ell-1}^{k+\ell-2}+C_{\ell-2}^{k+\ell-2})\\
&=C_{\ell}^{k+\ell-2}-C_{\ell-2}^{k+\ell-2}\\
&=\dim {\mathbb F}^{C_{\ell}^{k+\ell-2}-C_{\ell-2}^{k+\ell-2}}.
\end{align*}
Thus, $\Psi_{k,\ell,\overline{0}}$ is an injective linear map between vector spaces of the same dimension. Therefore, it is an isomorphism.
\end{proof}

\begin{cor}\label{YTE786jh} 
For a given $S\in {\mathbb F}^{C_{\ell}^{k+\ell-2}-C_{\ell-2}^{k+\ell-2}}$, $Y\in {\mathbb F}^{C_{\ell-1}^{k+\ell-1}}$ and  $A_{k}^{\ell}=A_{k}^{\ell}(0)$ in initial form  (see Definition \ref{obs36754}). 
Then the computational complexity of $\Psi_{k,\ell,Y}(S)$ is given by $$O(C_{\ell-1}^{k+\ell-1})+O((C_{\ell-2}^{k+\ell-1})^2).$$
\end{cor}
\begin{proof}
The table shows the computational complexity of calculating $\Psi_{k,\ell,Y}(S)$ for a given $S\in {\mathbb F}^{C_{\ell}^{k+\ell-2}-C_{\ell-2}^{k+\ell-2}}$,
$Y\in {\mathbb F}^{C_{\ell-1}^{k+\ell-1}}$ and  
$A_{k}^{\ell}$ in initial form  (see Definition \ref{obs36754}). 
\begin{center}
\begin{equation*}\label{table00001}
\begin{tabular}{| c | c | c |}
\hline
 $\Psi_{k,\ell,Y}(S)$ & {\bf TABLE}\quad I:  & \\
\hline
{\bf Operation}  & {\bf Comment} & {\bf Complexity}\\
 \hline
$AS^T$  & M‑SparseM & $O(C_{\ell}^{k+\ell-1})$\\
\hline
$ A_k^{\ell-1}Y_2^T$  & M‑SparseM & $O(C_{\ell}^{k+\ell-1})$\\
\hline
$A_k^{\ell-1}(AS^T)$  &  M‑SparseM & $O(C_{\ell}^{k+\ell-1})$\\
\hline
$P_1^T=\phi^{-1}[Y_1^T-A_k^{\ell-1}Y_2^T+(A_k^{\ell-1}A)S^T]$  & M‑SparseM & $O((C_{\ell-2}^{k+\ell-1})^2)$\\ 
\hline
$BP_1^T$    & M‑SparseM & $O(C_{\ell}^{k+\ell-1})$\\
\hline
$P_2^T=Y_2^T-BP_1^T-AS^T$   & Addition & $O(C_{\ell-1}^{k+\ell-1})$\\ 
\hline
\end{tabular}
\end{equation*}
\end{center}
\end{proof}
%

  
  




Using Definition \ref{HYU654}, we have the following proposition.
\begin{prop}
Let $k, \ell \geq 2$,  then 
\begin{align*}
\Psi_{k, \ell,\overline{0}}: {\mathbb F}^{C^{k+\ell-2}_{\ell}-C^{k+\ell-2}_{\ell-2}}&\longrightarrow {\EuScript C}_k^{\ell}\\
\Psi_{k,\ell,\overline{0}}(S)
&=S
\begin{bmatrix}
-\left[B\phi^{-1}(H_k^{\ell-1}A)+A\right]\\[2mm]
(\phi^{-1}H_k^{\ell-1}A)\\[2mm]
I_{C^{k+\ell-2}_{\ell}-C^{k+\ell-2}_{\ell-2}}
\end{bmatrix}^T.
\end{align*}                  
 is a systematic linear encoder.
\end{prop}
\begin{proof}
By the rank assumption, the dimension of $S_{(k,\ell,\overline{0})}$ is $C_\ell^{k+\ell-2}-C_{\ell-2}^{k+\ell-2},$
and equation \ref{casez340} shows that every codeword is uniquely determined by $S$.
By equation \ref{casez340} we have 
\begin{align*}
P_2^T(S)&=-BP_1^T-AS^T,\\
P_1^T(S)&=\phi^{-1}(H_k^{\ell-1}A)S^T,\\
S&\in{\mathbb F}^{C^{k+\ell-2}_{\ell}-C^{k+\ell-2}_{\ell-2}}.
\end{align*}
Substituting $P_1^T(S)$ into $P_2^T(S)$, we obtain
$$P_2^T(S)=-[B\phi^{-1}(H_k^{\ell-1}A)+A]S^T$$
Moreover, for $k,\ell \geq 2$, given the matrices $B$, $A$, $\phi$, and $H_k^{\ell-1}$ remain unchanged.
Therefore 
\begin{align*}
\Psi_{k,\ell,\overline{0}}(S)
&=S
\begin{bmatrix}
-\left[B\phi^{-1}(H_k^{\ell-1}A)+A\right]\\[2mm]
(\phi^{-1}H_k^{\ell-1}A)\\[2mm]
I_{C^{k+\ell-2}_{\ell}-C^{k+\ell-2}_{\ell-2}}
\end{bmatrix}^T
\end{align*}
\end{proof}
\begin{defn}\label{76ukl4327}
Let $k, \ell \geq 2$. Then 
\begin{align*}
\Psi_{k, \ell,\overline{0}}: {\mathbb F}^{C^{k+\ell-2}_{\ell}-C^{k+\ell-2}_{\ell-2}}&\longrightarrow {\EuScript C}_k^{\ell}\\
\Psi_{k,\ell,\overline{0}}(S)
&=S
\begin{bmatrix}
-\left[B\phi^{-1}(H_k^{\ell-1}A)+A\right]\\[2mm]
\phi^{-1}(H_k^{\ell-1}A)\\[2mm]
I_{C^{k+\ell-2}_{\ell}-C^{k+\ell-2}_{\ell-2}}
\end{bmatrix}^T.
\end{align*}                  
We call this encoder the \textit{Fractus encoder}.\\
We define the matrix
$$G_k^{\ell}= \begin{bmatrix}
-\left[B\phi^{-1}(H_k^{\ell-1}A)+A\right]\\[2mm]
\phi^{-1}(H_k^{\ell-1}A)\\[2mm]
I_{C^{k+\ell-2}_{\ell}-C^{k+\ell-2}_{\ell-2}}
\end{bmatrix}$$  
as the \emph{generator matrix} of the $(k,\ell)$-Fractus code. In particular, the \emph{Fractus encoder} is given by
$\Psi_{k,\ell,\overline {0}}(S)=S(G_k^{\ell})^T.$
Since the last block of $G_k^{\ell}$ is $I_K$, the generator matrix is systematic.
\end{defn}
By Corollary \ref{YTE786jh}, the computational complexity of the encoding is given by
\begin{equation}\label{RREt543}
O\left(C_{\ell-1}^{k+\ell-1}\right)
+
O\left(\left(C_{\ell-2}^{k+\ell-1}\right)^2\right).
\end{equation}
\subsubsection{An efficient algorithm for the encoding process}
\begin{center}
\hrulefill\\
{\bf Algorithm 1}\label{algor453}\\
\hrulefill\\
\end{center}
 {\bf Input}:    $S\in {\mathbb F}^{C_{\ell}^{k+\ell -2}-C_{\ell -2}^{k+\ell -2}}$\\
  {\bf Output}:  $(P_2, P_1, S)\in {\EuScript C}_k^{\ell}$\\ 
  
1) Compute $P_2^T=-\left[B\phi^{-1}(H_k^{\ell-1}A)+A\right]S^T$\\ 

2) Compute $P_1^T=\phi^{-1}(H_k^{\ell-1}A)S^T$ \\

3) Return $X=(P_2, P_1, S)$\\





\hrulefill\\

The following table shows the computational complexity where $P_2^T=-BP_1^T-AS^T$ is as 
in \ref{casez340} and $\operatorname{rank}(H_k^{\ell})=C_{\ell-1}^{k+\ell-1}$
\bigskip
\begin{equation*}\label{table00001}
\begin{tabular}{| c | c | c |}
\hline
$k\geq 2$,  $\ell\geq 2$  & \textbf{TABLE}\;   & \\
\hline
\textbf{Operation}  & \textbf{Comment} & \textbf{Complexity}\\
 \hline
$AS^T$  & Multiplication by sparse matrix & $O(C_{\ell}^{k+\ell-1})$\\
\hline
$H_k^{\ell-1}(AS^T)$  &  Multiplication by sparse matrix & $O(C_{\ell}^{k+\ell-1})$\\
\hline
$P_1^T=\phi^{-1}(H_k^{\ell-1}(AS^T))$  & Multiplication by "dense" matrix & $O((C_{\ell-2}^{k+\ell-1})^2)$\\ 
\hline
$BP_1^T$    & Multiplication by sparse matrix & $O(C_{\ell}^{k+\ell-1})$\\
\hline
$P_2^T=-BP_1^T-AS^T$   & Addition & $O((C_{\ell-1}^{k+\ell-1}))$\\ 
\hline
\end{tabular}\\
\end{equation*}
\bigskip
\begin{prop}\label{TGRE432}
Let $k, \ell\geq 2$  be  integers.\\  
Let $S=(S_1,S_2)\in {\mathbb F}^{C_{\ell -1}^{k+\ell -3}-C_{\ell -3}^{k+\ell -3}}\oplus {\mathbb F}^{C_{\ell }^{k+\ell -3}-C_{\ell -2}^{k+\ell -3}}$.
\begin{equation}\label{Tgr433}
\Psi_{k, \ell, \overline{0}}(S)=(\Psi_{k, \ell-1, \overline{0}}(S_1), \Psi_{k-1, \ell, \Psi_{k, \ell-1, \overline{0}}(S_1)}(S_2))
\end{equation}
\end{prop}

\begin{proof}
In this proposition, we solve the system of equations.
$$H_k^{\ell}\left[\begin{smallmatrix}
X_1^T \\
X_2^T 
\end{smallmatrix}\right]=\left[\begin{smallmatrix}
0^T \\
0^T 
\end{smallmatrix}\right]$$
using the $1$-fragmentation; see Definition \ref{obsyyt65}.
Which,  is equivalent to solving the following matrix system
\begin{equation}\label{apul6574ch54}
\begin{cases} 
   H_k^{\ell-1}X_1^T&=0^T\\
  H_{k-1}^{\ell}X_2^T&=-X_1^T
\end{cases}
\end{equation}
where $X_1=(X_{11}, X_{12})\in {\mathbb F}^{C^{k+\ell-3}_{\ell-2}}\oplus {\mathbb F}^{C^{k+\ell-3}_{\ell-1}}$ and $X_2=(X_{21}, X_{22})\in {\mathbb F}^{C^{k+\ell-3}_{\ell-1}}\oplus {\mathbb F}^{C^{k+\ell-3}_{\ell}}$.\\
Adapting the method of  theorem \ref{DfDfFr43765}, we solve the first matrix equation of \ref{apul6574ch54}, proceeding in the following manner:
$$
\begin{aligned}
\left[
\begin{array}{c|c}
H_k^{\ell-2} & 0 \\ \hline
I_{C_{\ell-2}^{k+\ell-3}} & H_{k-1}^{\ell-1}
\end{array}
\right]
\begin{bmatrix}
X_{11}^{T} \\
X_{12}^{T}
\end{bmatrix}
&\sim
\begin{bmatrix}
I_{C_{\ell-1}^{k+\ell-2}} & B_1 & A_1 \\
0 & -A_k^{\ell-2}B_1 & -A_k^{\ell-2}A_1
\end{bmatrix}
\begin{bmatrix}
P_{12}(S_1)^{T} \\
P_{11}(S_1)^{T} \\
S_1^{T}
\end{bmatrix} \\
&=
\begin{bmatrix}
0^{T} \\
0^{T}
\end{bmatrix}.
\end{aligned}
$$
with $B_1$ matrix of size $C_{\ell -2}^{k+\ell -3}\times C_{\ell -3}^{k+\ell -3}$ and $A_1$ matrix of size
$C_{\ell -2}^{k+\ell -3}\times( C_{\ell -1}^{k+\ell -3}-C_{\ell -3}^{k+\ell -3})$ such that $H_{k-1}^{\ell-1}=B_1\sqcup A_1$
and $\phi_1=-H_k^{\ell-1}B_1$ the invertible matrix of size 
$C_{\ell -3}^{k+\ell -3}\times C_{\ell -3}^{k+\ell -3}$, now for a vector $S_1$ whose values were chosen we have
\begin{align*}
P_{12}(S_1)^T&=-B_1P_{11}^T-A_1S_1^T\\
P_{11}(S_1)^T&=\phi_1^{-1}[(H_k^{\ell-2}A_1)S_1^T]\\
    S_1&\in {\mathbb F}^{C_{\ell -1}^{k+\ell -3}-C_{\ell -3}^{k+\ell -3}}
\end{align*}
Therefore
\begin{equation}\label{inc4r43}
 \Psi_{k, \ell-1, \overline{0}}(S_1)=(P_{12}(S_1), P_{11}(S_1), S_1)\in {\mathcal S}_{(k, \ell-1, \overline{0})}
 \end{equation}
So, if we do  $X_1=(X_{11}, X_{12})$,   where $X_{11}= P_{12}(S_1)$, \quad $X_{12}= (P_{11}(S_1), S_1^T)$.\\ 
To solve the second matrix equation in \ref{apul6574ch54},  we proceed as before (see Theorem \ref{DfDfFr43765}).

 Find $\left[\begin{smallmatrix}
X_{21}^T \\
X_{22}^T 
\end{smallmatrix}\right]$ such that 
\[
\left[
\begin{array}{c|c}
H_{k-1}^{\ell-1} & 0 \\ \hline
I_{C_{\ell-1}^{k+\ell-3}} & H_{k-2}^{\ell}
\end{array}
\right]
\left[
\begin{smallmatrix}
X_{21}^T \\
X_{22}^T
\end{smallmatrix}
\right]
=
-\begin{bmatrix}
P_{12}(S_1)^T \\
P_{11}(S_1)^T \\
S_1^T
\end{bmatrix}.
\]
This is equivalent to the following 
\[
\begin{bmatrix}
I_{C^{k+\ell-2}_{\ell-1}} & B & A \\
0 & -H_k^{\ell-1}B & -H_k^{\ell-1}A
\end{bmatrix}
\begin{bmatrix}
P_{22}(S_2)^T \\
P_{21}(S_2)^T \\
S_2^T
\end{bmatrix}
=-
\begin{bmatrix}
P_{12}(S_1)^T \\
P_{11}(S_1)^T \\
S_1^T
\end{bmatrix}.
\]
where $X_{11}=P_{12}(S_1)$ and $X_{12}=(P_{11}(S_1),S_1)$ are previously determined values, and $X_{21}=P_{22}(S_2)$ and $X_{22}=(P_{21}(S_2),S_2)$ are determined by $S_2$.\\
To do this, find the submatrix $B_2$ of size $C_{\ell -1}^{k+\ell -3}\times C_{\ell -2}^{k+\ell -3}$ and the submatrix $A_2$ of size
$C_{\ell -1}^{k+\ell -3}\times( C_{\ell }^{k+\ell -3}-C_{\ell -2}^{k+\ell -3})$ such that $H_{k-2}^{\ell}=B_2\sqcup A_2$ and $\phi_2=-H_{k-1}^{\ell-1}B_2$ the invertible matrix of size 
$C_{\ell -2}^{k+\ell -3}\times C_{\ell -2}^{k+\ell -3}$, now for a vector $S_2\in {\mathbb F}^{C_{\ell }^{k+\ell -3}-C_{\ell -2}^{k+\ell -3}}$ we have (see equation \ref{Fiell5654} in proof of theorem \ref{DfDfFr43765})
\begin{align*}
P_{21}(S_2)^T&=\phi_2^{-1}[P_{12}(S_1)^T-H_{k-1}^{\ell-1}\begin{bmatrix}
-P_{11}(S_1)^T \\
-S_1^T
\end{bmatrix} +(H_{k-1}^{\ell-1}A_2)S_2^T]\\
P_{22}(S_2)^T&=\begin{bmatrix}
-P_{11}(S_1)^T \\
-S_1^T
\end{bmatrix}-B_2P_{21}(S_2)^T-A_2S_2^T\\
        S_2&\in {\mathbb F}^{C_{\ell }^{k+\ell -3}-C_{\ell -2}^{k+\ell -3}}
\end{align*}
The solution  is $$\Psi_{k-1, \ell, \Psi_{k, \ell-1, \overline{0}}(S_1)}(S_2)=(P_{22}(S_2),  P_{21}(S_2), S_2)$$
Therefore
\begin{equation}\label{Tgr433}
\begin{aligned}
(\Psi_{k, \ell-1, \overline{0}}(S_1), \Psi_{k-1, \ell, \Psi_{k, \ell-1, \overline{0}}(S_1)}(S_2))&=(P_{12}(S_1),  P_{11}(S_1), S_1, P_{22}(S_2),  P_{21}(S_2), S_2 )\\                                                                                                          &=\Psi_{k, \ell, \overline{0}}((S_1, S_2))
\end{aligned}
\end{equation}
it is the required solution in ${\mathcal S}_{(k, \ell, \overline{0})}$ to $$S=(S_1, S_2)\in {\mathbb F}^{C_{\ell }^{k+\ell -2}-C_{\ell -2}^{k+\ell -2}},$$ 
such that
$$S_1\in {\mathbb F}^{C_{\ell -1}^{k+\ell -3}-C_{\ell -3}^{k+\ell -3}},$$ and 
$$S_2\in {\mathbb F}^{C_{\ell }^{k+\ell -3}-C_{\ell -2}^{k+\ell -3}}.$$ 
Then $$S=(S_1,S_2)\in {\mathbb F}^{C_{\ell -1}^{k+\ell -3}-C_{\ell -3}^{k+\ell -3}}\oplus {\mathbb F}^{C_{\ell }^{k+\ell -3}-C_{\ell -2}^{k+\ell -3}}={\mathbb F}^{C_{\ell }^{k+\ell -2}-C_{\ell -2}^{k+\ell -2}}$$.
\end{proof}
\begin{rem}\label{comp45565}
The computational complexity of $(\Psi_{k, \ell-1, \overline{0}}(S_1), \Psi_{k-1, \ell, \Psi_{k, \ell-1, \overline{0}}(S_1)}(S_2))$, is given by 
\begin{equation}\label{co45565}
\big[O(C_{\ell-2}^{k+\ell-2})+O((C_{\ell-3}^{k+\ell-2}))^2\big]+\big[O(C_{\ell-1}^{k+\ell-2})+O((C_{\ell-2}^{k+\ell-2}))^2\big].
\end{equation}
\end{rem}
\subsection{Decoding of ${\EuScript C}_k^{\ell}$}\label{231yytrfde}
LDPC decoding is an iterative process that uses graph-based algorithms to correct errors in data transmission.
A \textit{bipartite graph} is one in which the nodes can be partitioned into two classes and no edge connects two nodes from the same class.
\begin{lem}\label{Subm564544}
$H_k^2$ is the upper-left submatrix of $H_k^{\ell}$ for all $k, \ell \geq 2$
\end{lem}
\begin{proof}
Let $k\geq 2$ be an integer. By Definition \ref{definit000001}, for every $\ell \geq 2$ the matrix $H_k^{\ell-1}$ is the upper-left submatrix of $H_k^{\ell}$ (see Remark \ref{obs36754}).
Similarly, by the recursive definition, $H_k^{\ell-2}$ is the upper-left submatrix of $H_k^{\ell-1}$ (see Remark \ref{obs36754}). Applying transitivity repeatedly
$$H_k^2\unlhd H_k^3\unlhd \cdots \unlhd H_k^{\ell-1} \unlhd H_k^{\ell}$$
Therefore, $H_k^2$ is the upper-left submatrix of $H_k^{\ell}$.
\end{proof}
\begin{lem}\label{maths1333432}
If $k, \ell\geq 2$  then  $(H_k^{\ell})^TH_k^{\ell}$ is a matrix whose diagonal entries are constant and whose off-diagonal entries are integers less than or equal to 1
\end{lem}
\begin{proof}
The proof is by induction in $k$ and $\ell$.\\
The basis of induction is when $k=3$,  $\ell=2$  and $k=2$,  $\ell=3$ Thus, by a direct calculation, using Example \ref{gssto432} as a reference,
\bigskip\\
$(H_3^2)^TH_3^2=\begin{bmatrix}
2 & 1 & 1 & 1 & 1 & 0 \\
1 & 2 & 1 & 1 & 0 & 1 \\
1 & 1 & 2 & 0 & 1 & 1 \\
1 & 1 & 0 & 2 & 1 & 1 \\
1 & 0 & 1 & 1 & 2 & 1 \\
0 & 1 & 1 & 1 & 1 & 2
\end{bmatrix}$, \quad \quad
$(H_2^3)^TH_2^3=\begin{bmatrix}
3 & 1 & 1 & 1  \\
1 & 3 & 1 & 1  \\
1 & 1 & 3 & 1  \\
1 & 1 & 1 & 3  
\end{bmatrix}$\\
\bigskip\\
Our induction hypothesis states that, for each $2\leq \ell^{\prime} < \ell$ and for each $2\leq k^{\prime} < k$, the matrix  $(H_k^{\ell^{\prime}})^TH_k^{\ell^{\prime}}$ has diagonal entries equal  $\ell^{\prime}$  and off-diagonal entries less than or equal to 1. Similarly, the matrix  $(H_{k^{\prime}}^{\ell})^TH_{k^{\prime}}^{\ell}$ has diagonal entries equal to $\ell$ and off-diagonal entries less than or equal to 1.  Then for all $k$ and $\ell$ we have
\begin{align*}
(H_k^{\ell})^TH_k^{\ell}&= \left[ \begin{tabular}{c|c}
$(H_k^{\ell-1})^T$ & $I_{C_{\ell-1}^{k+\ell-2}} $ \cr\hline $0$  & $(H_{k-1}^{\ell})^T$
\end{tabular} \right] 
 \left[ \begin{tabular}{c|c}
$H_k^{\ell-1}$ & $ 0 $ \cr\hline $I_{C_{\ell -1}^{k+\ell -2}}$  & $H_{k-1}^{\ell}$
\end{tabular} \right] \\
&= \left[ \begin{tabular}{c|c}
$(H_k^{\ell-1})^TH_k^{\ell-1}+I_{C_{\ell -1}^{k+\ell -2}}$ & $ H_{k-1}^{\ell} $ \cr\hline $(H_{k-1}^{\ell})^T$  & $(H_{k-1}^{\ell})^TH_{k-1}^{\ell}$
\end{tabular} \right]  
\end{align*}
Thus, by the induction hypothesis,
$
(H_k^{\ell-1})^T H_k^{\ell-1}
+ I_{C_{\ell-1}^{k+\ell-2}}
$
is a matrix whose diagonal entries are all equal to $\ell$, while its off-diagonal entries are less than or equal to $1$.
Similarly, by the induction hypothesis,
$(H_{k-1}^{\ell})^T H_{k-1}^{\ell}$
is a matrix whose diagonal entries are all equal to $\ell$,
while its off-diagonal entries are less than or equal to $1$.
This completes the proof for
$(H_k^{\ell})^T H_k^{\ell}$.
\end{proof}
\begin{lem}\label{lemR}
If $k, \ell\geq 2$,  then the matrix
$$ H_k^{\ell} \sim \left[ \begin{tabular}{c|c}
$I_{C_{\ell -1}^{k+\ell -2}}$ & $ H_{k-1}^{\ell} $ \cr\hline $0$  & $-H_k^{\ell-1}H_{k-1}^{\ell}$
\end{tabular} \right] $$ 
\end{lem}
\begin{proof}
This follows from the matrix product below\\
$ \left[ \begin{tabular}{c|c}
$0$ & $I_{C_{\ell-1}^{k+\ell-2}} $ \cr\hline $I_{C_{\ell -2}^{k+\ell -2}}$  & $-H_k^{\ell-1}$
\end{tabular} \right] 
 \left[ \begin{tabular}{c|c}
$H_k^{\ell-1} $ & $ 0 $ \cr\hline $I_{C_{\ell -1}^{k+\ell -2}}$  & $H_{k-1}^{\ell}$
\end{tabular} \right] =
 \left[ \begin{tabular}{c|c}
$I_{C_{\ell -1}^{k+\ell -2}}$ & $ H_{k-1}^{\ell} $ \cr\hline $0$  & $-H_k^{\ell-1}H_{k-1}^{\ell}$
\end{tabular} \right] $ 
\end{proof}
\begin{cor}\label{wsewq435}
Assume that the field  ${\mathbb F}=GF(2^m)$.  Then the product $H_k^{\ell-1}H_{k-1}^{\ell}$ is the zero matrix, $k, \ell \geq 2$ 
\end{cor}
\begin{proof}
The proof proceeds by induction on $k$ and $\ell$.
If $\ell=1$, then $H_k^1=(1,\ldots,1)$, and $H_{k-1}^2$ has two ones in each column. Therefore,
$$H_k^1H_{k-1}^2=(1,\ldots,1)H_{k-1}^2=(2,\ldots,2)\cong(0,\ldots,0).$$
and if $k=1$, then $H_2^{\ell-1}H_1^{\ell}$ has two ones in each row and $H_1^{\ell}=\begin{bmatrix}
1 \\
\vdots \\
1 
\end{bmatrix}$.
 Therefore

$$ H_2^{\ell-1}H_1^{\ell} \sim \left[ \begin{tabular}{c|c}
$ H_2^{\ell-2}$ & $0$ \cr\hline $I_{C_{\ell-2}^{\ell-1}}$  & $H_1^{\ell-1}$
\end{tabular} \right] 
\begin{bmatrix}
1 \\
\vdots \\
1 
\end{bmatrix}
=\begin{bmatrix}
2 \\
\vdots \\
2 
\end{bmatrix}
\cong \begin{bmatrix}
0 \\
\vdots \\
0 
\end{bmatrix}$$
Our induction hypothesis is the following:\\
For all $2\leq k^{\prime}<k$ and $2\leq \ell^{\prime}<\ell$ we have,  $$H_k^{\ell^{\prime}-1}H_{k-1}^{\ell^{\prime}}\cong \overline{0}\mod 2$$
and $$ H_{k^{\prime}}^{\ell-1}H_{k^{\prime}-1}^{\ell}\cong \overline{0}\mod 2.$$ 
Applying the induction hypothesis, we have,\\
\begin{center}
$ H_k^{\ell-1}H_{k-1}^{\ell} = \left[ \begin{tabular}{c|c}
$H_k^{\ell-2}H_{k-1}^{\ell-1}$ & $0$ \cr\hline $2H_{k-1}^{\ell-1}$  & $H_{k-1}^{\ell-1}H_{k-2}^{\ell}$
\end{tabular} \right]\cong $
$\begin{bmatrix}
0 & 0 \\
0 & 0 
\end{bmatrix}$\\
\end{center}
\end{proof}
\begin{cor}\label{wsewq435777}
Assume that the field  ${\mathbb F}=GF(2^m)$. Then $$ H_k^{\ell} \sim \left[ \begin{tabular}{c|c}
$I_{C_{\ell -1}^{k+\ell -2}}$ & $ H_{k-1}^{\ell} $ \cr\hline $0$  & $0$
\end{tabular} \right] $$ 
\end{cor}
\begin{proof}
The proof follows from the lemma \ref{lemR} and corollary \ref{wsewq435}.
\end{proof}
%

%

%
\begin{defn}\label{56543lkj}
Let $H=(h_{ij})$ be an $m\times n$ parity-check matrix of a linear code. The \textit{Tanner graph} associated with $H$ is the bipartite graph $G=(R, C, E)$, where 
$R=\{r_1, \ldots, r_n \}$ is the set of \textit{variable nodes}, $C=\{c_1, \ldots, c_m \}$ is the set of \textit{check nodes}, and an edge $(r_i, c_j)\in E$ exists if and only if $h_{ij}=1$\\
\end{defn}
\begin{defn}\label{gird4321}
A cycle of length $t$  in a Tanner graph is a path comprised of $t$ edges from a node back to the same node. Obviously, the shortest possible
 cycle in any Tanner graph is four. The length of the shortest cycles in an LDPC code ${\EuScript C}$ is called the \textit{ girth} and denote by
 $g({\EuScript C})$. \\
 \end{defn}
Cycles, especially short cycles, leave a bad effect on the performance of LDPC decoders. Because
they affect the independence of the extrinsic information which exchanged in iterative decoding.
\begin{prop}\label{6cicl35555}
The Tanner Graph $T_{{\EuScript C}_k^2}$ have 6-cycles
\end{prop}
\begin{proof}
The proof proceeds by induction on  $k$. We show that the matrix $H_k^2$ contains a 6-cycle pattern as in \cite[page 930 fig. 2]{1.1}
$$H_2^2=\begin{bmatrix}
\bullet & \bullet & 0  \\
\bullet & 0 & \bullet \\
0 & \bullet & \bullet \\
\end{bmatrix}$$
$$H_3^2=\begin{bmatrix}
\bullet & \bullet & 1 & 0 & 0 & 0 \\
\bullet &     0     & 0 & \bullet & 1 &  0\\
0 & \bullet & 0 & \bullet & 0 &  1\\
0 & 0 & 1 & 0 & 1 &  1 \\
\end{bmatrix}$$
$$H_4^2=\begin{bmatrix}
\bullet & \bullet & 1 & 1 & 0 & 0 & 0 & 0 & 0 & 0\\
\bullet &    0     &  0 & 0 & \bullet & 1 &  1 & 0 & 0 & 0\\
        0 & \bullet & 0 & 0 & \bullet & 0 &  0 & 1 & 1 &  0\\
        0 & 0 & 1 & 0 & 0 & 1 &  0 & 1 & 0 &  1\\
0 & 0 & 0 & 1 & 0 & 0 &  1 & 0 & 1 &  1
\end{bmatrix}$$
$$ \vdots$$
$$H_k^2=\begin{bmatrix}
\bullet & \bullet & \cdots & 1        &     0    &    \cdots  &  0     &  \cdots & 0\\
\bullet &   0       & \cdots & 0        & \bullet &   \cdots  & 1     &   \cdots & 0\\
0         & \bullet & \cdots & 0        & \bullet &    \cdots & 0      &  \cdots & 0\\
\vdots & \vdots & \cdots & \vdots & \vdots &  \vdots &  \cdots & \vdots & \vdots\\
0 & 0 & \cdots  &   1       &   0      & \cdots  &     0     &  \cdots &     1
\end{bmatrix}$$\\
and so we have proven it inductively.
\end{proof}
It is well known the Tanner-Gallager criterion regarding 4-cycles, which states: If $H^T H$ has off-diagonal entries equal to 0 or 1, then the Tanner graph has no cycles of length 4 (see \cite{0.1}, \cite{3}).
\begin{example}
The Tanner graph  $T_{{\EuScript C}_3^2}$ 
where $H_3^2$ is the parity check matrix
$$H_3^2=\begin{bmatrix}
1 & 1 & 1 & 0 & 0 & 0 \\
1 & 0 & 0 & 1 & 1 &  0\\
0 & 1 & 0 & 1 & 0 &  1\\
0 & 0 & 1 & 0 & 1 &  1 \\
\end{bmatrix}$$\\
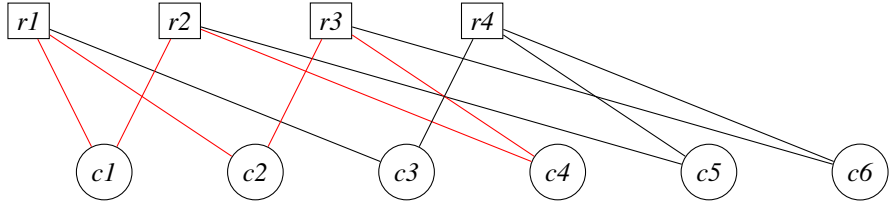
\begin{center}
\begin{tikzpicture}

\node[draw,rectangle] (r1) at (0, 2) {r1};
\node[draw,rectangle] (r2) at (2, 2) {r2};
\node[draw,rectangle] (r3) at (4, 2) {r3};
\node[draw,rectangle] (r4) at (6, 2) {r4};

\node[draw,circle] (c1) at (1, 0) {c1};
\node[draw,circle] (c2) at (3, 0) {c2};
\node[draw,circle] (c3) at (5, 0) {c3};
\node[draw,circle] (c4) at (7, 0) {c4};
\node[draw,circle] (c5) at (9, 0) {c5};
\node[draw,circle] (c6) at (11, 0) {c6};

\draw[red] (r1) -- (c1);
\draw[red] (r1) -- (c2);
\draw (r1) -- (c3);

\draw[red] (r2) -- (c1);
\draw[red] (r2) -- (c4);
\draw (r2) -- (c5);

\draw[red] (r3) -- (c2);
\draw[red] (r3) -- (c4);
\draw (r3) -- (c6);

\draw (r4) -- (c3);
\draw (r4) -- (c5);
\draw (r4) -- (c6);

\end{tikzpicture}
\end{center}
The graph shows a cycle of length 6 in red. \\ 
$$(H_3^2)^T H_3^2=\begin{bmatrix}
2 & 1 & 1 & 1 & 1 & 0 \\
1 & 2 & 1 & 1 & 0 &  1\\
1 & 1 & 2 & 0 & 1 &  1\\
1 & 1 & 0 & 2 & 1 &  1 \\
1 & 0 & 1 & 1 & 2 &  1\\
0 & 1 & 1 & 1 & 1 & 2 
\end{bmatrix}$$\\
By the Tanner–Gallager criterion, the Tanner graph in this example is free of 4-cycles and has girth 6.
\end{example}
\begin{thm}\label{gfbd3233}
If $2\leq \ell < k$,  then  Tanner graph $T_{{\EuScript C}_k^{\ell}}$ of ${\EuScript C}_k^{\ell}$ has girth 6.
\end{thm}
\begin{proof}
The parity check matrix of ${\EuScript C}_k^{\ell}$ is $H_k^{\ell}$. By Lemma \ref{maths1333432}  the matrix $(H_k^{\ell})^TH_k^{\ell}$ has off-diagonal entries less than or equal to 1, so the Tanner graph $T_{{\EuScript C}_k^{\ell}}$ contains no 4-cycles.  By lemma \ref{Subm564544} $H_k^2$ is the upper-left submatrix of  $H_k^{\ell}$. It follows from Lemma \ref{6cicl35555} that the Tanner graph $T_{{\EuScript C}_k^{\ell}}$  has girth 6.
\end{proof}
\begin{cor}\label{HYt5654}
For each $k, \ell\geq 2$. Then  $d_{{\EuScript C}_k^{\ell}}\geq \ell+1$ 
\end{cor}
\begin{proof}
By theorem \ref{gfbd3233} the girth of Tanner graph of ${\EuScript C}_k^{\ell}$ is $g=6$ and by Tanners Tree Bound $d_{{\EuScript C}_k^{\ell}}\geq \ell+1$
\end{proof}
\subsection{$(k,\ell)$-\textit{Fractus Codes over the finite field } $GF(2^m)$} 
In this part, we assume that ${\mathbb F}=GF(2^m)$ as base field. By Corollary \ref{wsewq435777}, we have
$$H_{k+1}^{\ell}\sim
\left[I_{C_{\ell-1}^{k+\ell-1}} \mid H_k^{\ell}\right].$$
\begin{defn}\label{Akl4327432}
Let $k,\ell\geq 2$ be integers. We define ${\EuScript C}_{(I_{C_{\ell-1}^{k+\ell-1}} |H_k^{\ell})}$ the \textit{truncated Fractus code} over $GF(2^m)$ as the LDPC code whose parity-check matrix is
$H=\left[I_{C_{\ell-1}^{k+\ell-1}}\mid H_k^{\ell}\right].$
\end{defn}
The code has 
\begin{enumerate}
\item length  $N=C_{\ell-1}^{k+\ell-1}+C_{\ell}^{k+\ell-1}=C_{\ell}^{k+\ell}$,
\item dimension $K=N-rank H_k^{\ell}=C_{\ell}^{k+\ell}-C_{\ell-1}^{k+\ell-1}=C_{\ell-1}^{k+\ell-1}$
\item rate  $r=\frac{C_{\ell-1}^{k+\ell-1}}{C_{\ell}^{k+\ell}}=\frac{\ell}{k+\ell}$.
\end{enumerate}
The following systematic encoding for codes with parity-check matrix of the form
$H=\left(I_{C_{\ell-1}^{k+\ell-1}}\mid H_k^{\ell}\right)$
is standard. We include it here for completeness.
\begin{center}
\hrulefill\\
{\bf Algorithm 2}\label{31Algrt4443}\\
\hrulefill\\
\end{center}
 {\bf Input}: $S\in {\mathbb F}^{C_{\ell}^{k+\ell-1}}$ \text{arbitrary vector}.\\
 {\bf Output}: $x=(S\mid S(H_k^{\ell})^T)$ a word of ${\EuScript C}_{(I_{C_{\ell-1}^{k+\ell-1}} \mid H_k^{\ell})}$\\ 
a) Estimate $H_k^{\ell}S^T$,\\
b) Return $x(S)=S
\begin{bmatrix}
I_{C_{\ell-1}^{k+\ell-1}} \\
H_k^{\ell} 
\end{bmatrix}^T.$\\

\hrulefill\\
\begin{rem}\label{DErf443}
Note that the computational complexity of previous Algorithm is dominated by the computation of $S(H_k^{\ell})^T$, a sparse matrix-vector product with computational complexity $O(C_{\ell-1}^{k+\ell-1})$.
\end{rem}
\begin{lem}\label{nheduc543}
The Tanner Graph of ${\EuScript C}_{(I_{C_{\ell-1}^{k+\ell-2}} \mid H_k^{\ell})}$ has girth is 6.                                                                               
\end{lem}
\begin{proof}
Consider the following sequence of submatrices $H_k^2\unlhd (I_{C_{\ell-1}^{k+\ell-2}}\mid H_k^{\ell})\unlhd H_{k+1}^{\ell}$
then by theorem \ref{gfbd3233} we have  $6=g(H_{k+1}^{\ell})\leq g((I_{C_{\ell-1}^{k+\ell-2}}| H_k^{\ell}))\leq g(H_k^2)=6$  hence
$g({\EuScript C}_{(I_{C_{\ell-1}^{k+\ell-2}} \mid H_k^{\ell})})=6$.
\end{proof}
\begin{prop}\label{Dm5654c3}
The  LDPC-code  ${\EuScript C}_{(I_{C_{\ell-1}^{k+\ell-1}} |H_k^{\ell})}$ satisfies $d_{\min}({\EuScript C}_{(I_{C_{\ell-1}^{k+\ell-1}} \mid H_k^{\ell})})=\ell+1$
\end{prop}
\begin{proof}
By lemma \ref{nheduc543} and Tanner´s Tree Bound we have $\ell+1\leq d_{ \min}({\EuScript C}_{(I_{C_{\ell-1}^{k+\ell-1}} \mid H_k^{\ell})})$.\\
For any codeword, we can parametrize it in the form  $$x(S)=(S\mid S(H_k^{\ell})^T)$$ for $S\in {\mathbb F}^{C_{\ell}^{k+\ell-1}}$
and its weight is $$d_{ \min}({\EuScript C}_{(I_{C_{\ell-1}^{k+\ell-1}} \mid H_k^{\ell})})=\min_{v\neq0}(wt(S)+wt(S(H_k^{\ell})^T))$$
If $wt(S)=1$ then $S=e_j$, the $j$-th canonical vector, and $S(H_k^{\ell})^T$ is equal to the row $j$ of $(H_k^{\ell})^T$, each row has weight $\ell$. This implies that 
$$wt(x)=1+\ell$$ and this means that $d_{ \min}({\EuScript C}_{(I_{C_{\ell-1}^{k+\ell-1}} \mid H_k^{\ell})})\leq \ell+1$
and $\ell+1\leq d_{ \min}({\EuScript C}_{(I_{C_{\ell-1}^{k+\ell-1}} \mid H_k^{\ell})})\leq \ell+1$
\end{proof}
\begin{example}
Consider the matrix $(I_{10}|A^3_3)$  where
$$A^3_3= \begin{bmatrix}
1 & 1 & 1 & 0 & 0 & 0 & 0 & 0 & 0& 0 \\
1 & 0 & 0 & 1 & 1 & 0 & 0 & 0 & 0& 0 \\
0 & 1 & 0 & 1 & 0 & 1 & 0 & 0 & 0& 0\\
0 & 0 & 1 & 0 & 1 & 1 & 0 & 0 & 0& 0\\
0 & 0 & 0 & 0 & 0 & 0 & 1 & 1 & 0& 0\\
0 & 1 & 0 & 0 & 0 & 0 & 1 & 0 & 1& 0\\
0 & 0 & 1 & 0 & 0 & 0 & 0 & 1 & 1& 0\\
0 & 0 & 0 & 1 & 0 & 0 & 1 & 0 & 0 & 1\\
0 & 0 & 0 & 0 & 1 & 0 & 0 & 1 & 0 & 1\\
0 & 0 & 0 & 0 & 0 & 1 & 0 & 0 & 1 & 1\\
\end{bmatrix}$$
then  ${\EuScript C}_{(I_{10}|A^3_3)}$ is a $(20, 10, 4)$-LDPC code.
\end{example}
\begin{lem} \label{Hjkuy543} 
The Tanner graph of ${\EuScript C}_k^{\ell}$ satisfies \[ T_{{\EuScript C}_k^{\ell}} = T_{{\EuScript C}_k^{\ell-1}} \sqcup T_{{\EuScript C}_{(I_{C_{\ell-1}^{k+\ell-2}} \mid H_{k-1}^{\ell})}}. \] \end{lem}
\begin{proof}
Using the notation introduced in Definition \ref{56543lkj} for variable nodes, check nodes and edges,  we have
$$T_{{\EuScript C}_k^{\ell}}=(R_{{\EuScript C}_k^{\ell}}, C_{{\EuScript C}_k^{\ell}}, E_{{\EuScript C}_k^{\ell}})$$
$$T_{{\EuScript C}_k^{\ell-1}}=(R_{{\EuScript C}_k^{\ell-1}}, C_{{\EuScript C}_k^{\ell-1}}, E_{{\EuScript C}_k^{\ell-1}})$$ 
$$T_{{\EuScript C}_{(I_{C_{\ell-1}^{k+\ell-2}} \mid H_{k-1}^{\ell})}}=(R_{(I_{C_{\ell-1}^{k+\ell-2}} \mid H_{k-1}^{\ell})}, 
C_{(I_{C_{\ell-1}^{k+\ell-2}} \mid H_{k-1}^{\ell})}, E_{(I_{C_{\ell-1}^{k+\ell-2}} \mid H_{k-1}^{\ell})})$$
are the Tanner graphs respectively. 
By \ref{m4tr1x67532}, it is easy to see that the submatrices $(H_k^{\ell-1} \mid 0)$ and $(I_{C_{\ell-1}^{k+\ell-2}} \mid H_{k-1}^{\ell})$ have sizes $C_{\ell-2}^{k+\ell-2}\times C_{\ell}^{k+\ell-1}$ and $C_{\ell-1}^{k+\ell-2}\times C_{\ell}^{k+\ell-1}$, respectively. They are disjoint and complementary submatrices of $H_k^{\ell}$.
 Then we can conclude the following
$$R_{{\EuScript C}_k^{\ell}}= R_{{\EuScript C}_k^{\ell-1}}\cup R_{(I_{C_{\ell-1}^{k+\ell-2}} \mid H_{k-1}^{\ell})}$$
and $$E_{{\EuScript C}_k^{\ell}}=E_{{\EuScript C}_k^{\ell-1}}\cup E_{(I_{C_{\ell-1}^{k+\ell-2}} \mid H_{k-1}^{\ell})}$$
Also, given that $T_{{\EuScript C}_k^{\ell-1}}=T_{{\EuScript C}_(H_k^{\ell-1}|0)}$ and $|C_{{\EuScript C}_k^{\ell}}|=|C_{{\EuScript C}_k^{\ell-1}}|=|C_{(I_{C_{\ell-1}^{k+\ell-2}} \mid H_{k-1}^{\ell})}|=C_{\ell}^{k+\ell-1}$
we have to
$$C_{{\EuScript C}_k^{\ell}}=C_{{\EuScript C}_k^{\ell-1}}=C_{(I_{C_{\ell-1}^{k+\ell-2}} \mid H_{k-1}^{\ell})}$$
 Therefore, 
$T_{{\EuScript C}_k^{\ell}}
=
T_{{\EuScript C}_k^{\ell-1}}
\sqcup
T_{{\EuScript C}_{(
I_{C_{\ell-1}^{k+\ell-2}}
\mid H_{k-1}^{\ell}
)}}.$
\end{proof}
\begin{prop}\label{65tfr2BBBII}
Let $k, \ell \geq 2$. Then the Tanner graph $T_{{\EuScript C}_k^{\ell}}$ of fractus codes satisfies 
\[
T_{{\EuScript C}_k^{\ell}}
=
T_{{\EuScript C}_{{\EuScript C}_k^1}}
\sqcup
T_{{\EuScript C}_{(I_{C_1^{k}} \mid H_{k-1}^2)}}
\sqcup
\ldots
\sqcup
T_{{\EuScript C}_{(I_{C_{\ell-2}^{k+\ell-3}} \mid H_{k-1}^{\ell-1})}}
\sqcup
T_{{\EuScript C}_{(I_{C_{\ell-1}^{k+\ell-2}} \mid H_{k-1}^{\ell})}}
\]\end{prop}
\begin{proof}
Of lemma \ref{Hjkuy543} we have that 
\[
T_{{\EuScript C}_k^{\ell}}
=
T_{{\EuScript C}_k^{\ell-1}}
\sqcup
T_{{\EuScript C}_{(I_{C_{\ell-1}^{k+\ell-2}}\mid H_{k-1}^{\ell})}}
\]
and $T_{{\EuScript C}_k^{\ell-1}}
=
T_{{\EuScript C}_k^{\ell-2}}
\sqcup
T_{{\EuScript C}_{(I_{C_{\ell-2}^{k+\ell-3}} \mid H_{k-1}^{\ell-1})}}$
we easily conclude 
$$
T_{{\EuScript C}_k^{\ell}}
=
T_{{\EuScript C}_k^{\ell-2}}
\sqcup
T_{{\EuScript C}_{(I_{C_{\ell-2}^{k+\ell-3}} \mid H_{k-1}^{\ell-1})}}
\sqcup
T_{{\EuScript C}_{(I_{C_{\ell-1}^{k+\ell-2}} \mid H_{k-1}^{\ell})}}
$$
$$
T_{{\EuScript C}_k^{\ell}}
=
T_{{\EuScript C}_k^{\ell-3}}
\sqcup
T_{{\EuScript C}_{(I_{C_{\ell-3}^{k+\ell-4}} \mid H_{k-1}^{\ell-2})}}
\sqcup
T_{{\EuScript C}_{(I_{C_{\ell-2}^{k+\ell-3}} \mid H_{k-1}^{\ell-1})}}
\sqcup
T_{{\EuScript C}_{(I_{C_{\ell-1}^{k+\ell-2}} \mid H_{k-1}^{\ell})}}
$$
If we continue in this way, we can conclude that
\[
T_{{\EuScript C}_k^{\ell}}
=
T_{{\EuScript C}_{{\EuScript C}_k^1}}
\sqcup
T_{{\EuScript C}_{(I_{C_1^{k}} \mid H_{k-1}^2)}}
\sqcup
\ldots
\sqcup
T_{{\EuScript C}_{(I_{C_{\ell-2}^{k+\ell-3}} \mid H_{k-1}^{\ell-1})}}
\sqcup
T_{{\EuScript C}_{(I_{C_{\ell-1}^{k+\ell-2}} \mid H_{k-1}^{\ell})}}
\]
\end{proof}
\begin{example}
Let us consider the $(3,6)$-fractus code ${\EuScript C}_6^3$ associated with the fractus matrix $A_6^3$. This code is a $[35, 3, 6]$ code with rate $0.5$, and its encoding complexity is $O(56)+O(49)$. 

\begin{center}
Table 1: Parameters of example codes with rate=1/2.\\
\bigskip
\begin{tabular}{l c c c c c}
\hline
Code & $k$  & $\ell$ & $N$ &  $r$ &  $g$\\
\hline
${\EuScript C}_6^3$ & 6 & 3 & 56   & 1/2  &    6 \\
Mac Kay's & 6 & 3 & 1008   & 1/2  &    6 \\
Tanner QC & 6 & 3 & 1002  & 1/2  &    8 \\
${\EuScript C}_{2\ell}^{\ell}$ & $2\ell$ & $\ell$ & $C_{\ell-1}^{k+\ell-2}$ & 1/2 & 6\\

\hline
\end{tabular}
\end{center}
\bigskip
To demonstrate the error-correction performance, we constructed the rate-$1/2$ Fractus code $C$. For comparison, we also constructed two classes of codes: Tanner's QC codes and MacKay's random codes. The MacKay random codes and Tanner's QC codes were obtained from \cite{2.0}. Both the MacKay random codes and Tanner's QC codes used for comparison are high-performance codes. The parameters of the selected codes are presented in Table 1.
Table 1 lists the typical values of the row weight $k$, column weight $j$, code length $N$, code rate $R$, and girth $g$ for the three classes of codes.
\end{example}
\begin{example}
Let ${\mathbb F}=GF(q)$ and suppose that $q > \ell$. We define a family $\{C_{2\ell}^{\ell} : \ell\geq 2\}$ of  $(2\ell, \ell)$-Fractus codes with parameters $[C_{\ell}^{3\ell-1},\ell,d]$ and  
constant rate $r=1-\frac{\ell}{2\ell}=1/2$. By corollary \ref{HYt5654} the minimum distance $d$ satisfies $ d>\ell+1$. These codes have constant girth equal to 6, 
$\upsilon=\text{number of 1´s in the matrix} \; A_{2\ell}^{\ell}$

and density of ones given by $\delta:=\frac{2\ell}{C_{\ell}^{3\ell-1}}$
\begin{center}
Table 2: Parameters of family  with rate=1/2.\\
\bigskip
\begin{tabular}{l c c c c c}
\hline
$\ell$ & $2\ell$  & $\ell+1< d$ & $N$ &  $\upsilon$ & $\delta$ \\
\hline
2 & 10 & 3 & 10   &     20 & 0.40   \\
3& 6 & 4 & 56   &    168  &  0.1071  \\
4& 8 & 5 & 330  &   1320 &0.0242   \\
5& 10 & 6 & 1002  &  10,010  & $4.9\times 10^{-3}$  \\
\hline
\end{tabular}
\end{center}
\bigskip
\end{example}
\section{Fragmentation of codewords}
As stated in Definition~\ref{obsyyt65}, the $r$-fragmentation $H_k^{\ell}(r)$ is a block representation of the same matrix $H_k^{\ell}$, for $0 \leq r \leq \min\{k-1,\ell-1\}$. Thus, we have
$$H_k^{\ell}=H_k^{\ell}(0)=H_k^{\ell}(1)=\cdots =H_k^{\ell}(r-1)=H_k^{\ell}(\min\{k-1, \ell-1 \}) $$
\begin{prop}\label{TrePP3d}
Let   $k,\ell \geq 2$ be integers, $0\leq r \leq \min\{k-1, \ell-1 \} $ and let $H_k^{\ell}(r)$ be the $r$-fragmentation  of $H_k^{\ell}$. Then $H_k^{\ell}(r)$ admits a block lower triangular representation such that each block on the main diagonal is a $(k^{\prime},\ell^{\prime})$-fractus submatrix arising from the recursive contruction of  $H_k^{\ell}$.  There are exactly $2^r$ fractus submatrices on the main diagonal, and every other nonzero block is an identity matrix.
\end{prop}
\begin{proof}
We proceed by induction on $r$.
For $r=0$, by definition,
$H_k^\ell(0)=H_k^\ell.$
Thus, $H_k^\ell(0)$ itself is a $(k,\ell)$-fractus matrix and constitutes the unique block on the main diagonal. Hence, there is exactly $2^0=1$
fractus submatrix on the main diagonal, and there are no other nonzero blocks. Therefore, the statement holds for $r=0$.

Now suppose that the statement holds for $r-1$, where
$1\leq r\leq \min\{k-1,\ell-1\}.$
By the recursive definition of the $r$-fragmentation,
$$H_k^{\ell}(r)=\left[
\begin{tabular}{c|c}
$H_k^{\ell-1}(r-1)$ & $0$ \cr\hline
$I_{C_{\ell-1}^{k+\ell-2}}$ & $H_{k-1}^{\ell}(r-1)$
\end{tabular}
\right].$$
By the induction hypothesis, $H_k^{\ell-1}(r-1)$ admits a block lower triangular representation whose main diagonal consists of exactly $2^{r-1}$ fractus submatrices of $H_k^{\ell-1}$, and whose other nonzero blocks are identity matrices. Likewise, $H_{k-1}^{\ell}(r-1)$ admits a block lower triangular representation whose main diagonal consists of exactly $2^{r-1}$ fractus submatrices of $H_{k-1}^{\ell}$, and whose other nonzero blocks are identity matrices.

Substituting these two block representations into the recursive definition of $H_k^\ell(r)$, we obtain a block lower triangular representation of $H_k^\ell(r)$. Its main diagonal is obtained by concatenating the main diagonals of $H_k^{\ell-1}(r-1)$ and $H_{k-1}^{\ell}(r-1)$. Consequently, the number of fractus submatrices on the main diagonal is $2^{r-1}+2^{r-1}=2^r.$
Moreover, every fractus submatrix occurring on the main diagonal is a $(k^{\prime}, \ell^{\prime})$-fractus submatrix of $H_k^\ell$. The only additional nonzero blocks introduced by the recursive construction are identity matrices, namely
$I_{C_{\ell-1}^{k+\ell-2}},$ together with the identity blocks inherited from the two $(r-1)$-fragmented forms. Hence, every nonzero block outside the main diagonal is an identity matrix. Therefore, $H_k^\ell(r)$ admits a block lower triangular representation satisfying all the required properties. This completes the induction.
\end{proof}
\begin{defn}\label{obs36754} 
Let $k,\ell\geq 2$ be integers, let $0\leq r\leq\min\{k-1,\ell-1\}$, and let $H_k^\ell$ be a $(k,\ell)$-fractus matrix. The main diagonal of $H_k^\ell(r)$, the $r$-fragmentation of $H_k^\ell$, viewed as a vector of length $2^r$ and ordered according to the $r$-fragmentation, is called the $(k,\ell,r)$-\textit{mosaic} and is denoted by ${\mathfrak M}(k,\ell,r)$. The elements of ${\mathfrak M}(k,\ell,r)$ are referred to as \textit{tessellations}. We say that a tessellation has \textit{weight} $p$ if it occurs $p$ times in the mosaic ${\mathfrak M}(k,\ell,r)$. We denote by ${\mathfrak A}(k,\ell,r)=\{ H_k^{\ell-r}, H_{k-1}^{\ell-(r+1)},\ldots,  H_{k-(r+1)}^{\ell-1}, H_{k-r}^{\ell}\}$ 
the support of the mosaic ${\mathfrak M}(k,\ell,r)$ and refer to it as the \textit{Atlas of tessellations}.
The \textit{mosaic weight vector}, denoted by $wt({\mathfrak A}(k,\ell,r))$, is a vector whose entries correspond to the weights of the elements of the Atlas of tessellations, ordered according to the decreasing sequence of indices
$k>k-1>\cdots>k-r.$
\end{defn}
\begin{lem}
The weights of the tessellations in the mosaic ${\mathfrak M}(k,\ell,r)$ associated with the $r$-fragmentation $H_k^\ell(r)$ follow a pattern analogous to that of Newton's binomial
coefficients. In particular, the sequence appears to satisfy Pascal's recurrence and therefore exhibits a binomial-type combinatorial structure.
\end{lem}
\begin{proof}
The proof proceeds by induction on $r$, where $0 \leq r \leq \min\{k-1,\ell-1\}$. To illustrate the behavior of the weight distribution of the tessellations, we use a broader base case than is required for the induction.
Based on example \ref{GTR543hu}, we have that the mosaics are given by:
\begin{align*}
{\mathfrak M}{(k, \ell, 0)}&=(H_k^{\ell})\\
{\mathfrak M}{(k, \ell, 1)}&=(H_k^{\ell-1},  H_{k-1}^{\ell})\\
{\mathfrak M}{(k, \ell, 2)}&=(H_k^{\ell-2}, H_{k-1}^{\ell-1}, H_{k-1}^{\ell-1}, H_{k-2}^{\ell})\\
{\mathfrak M}{(k, \ell, 3)}&=(H_k^{\ell-3}, H_{k-1}^{\ell-2}, H_{k-1}^{\ell-2}, H_{k-2}^{\ell-1}, H_{k-1}^{\ell-2}, H_{k-2}^{\ell-1}, H_{k-2}^{\ell-1}, H_{k-3}^{\ell})
\end{align*}
The mosaic weight vector,
\begin{table}[h]
  \centering
  \begin{tabular}{c|c}
    r &  ${wt(\mathfrak A}(k,\ell,r))$ \\ \hline
   0 &  1  \\ \hline
   1 &  (1,1) \\ \hline
   2 & (1,2,1) \\ \hline
   3 &  (1,3,3,1) \\ \hline
   4 & (1,4,6,4,1) \\ \hline
  \end{tabular}
  \caption{mosaic weight vector}
  \label{tab:mi_tabla}
\end{table}
By definition of mosaic (see,  \ref{obs36754}) $${\mathfrak M}{(k, \ell, r)}=({\mathfrak M}{(k, \ell-1, r-1)}\mid {\mathfrak M}{(k-1, \ell, r-1)})$$
then our induction base is given by a pattern analogous to that of Newton's binomial coefficients
so $wt({\mathfrak M}{(k, \ell-1, r-1)})$ and 
$wt( {\mathfrak M}{(k-1, \ell, r-1)})=\sum_{i=0}^{r-1}C_{i=0}^{r-1}$
\begin{align*}
wt({\mathfrak M}{(k, \ell-1, r-1)})+wt( {\mathfrak M}{(k-1, \ell, r-1)})&=\sum_{i=1}^rC_{i-1}^{r-1}+\sum_{i=1}^{r-1}C_i^{r-1}\\
                                                                                                       &=C_0^{r-1}+\big[\sum_{i=2}^{r-1}C_{i-1}^{r-1}+\sum_{i=1}^{r-2}C_i^{r-1}\big]+C_{r-1}^{r-1}\\    
                                                                                                       &=1+\big[\sum_{i=2}^{r-2}C_{i-1}^{r-1}+C_i^{r-1}\big]+1\\ 
                                                                                                       &=C_0^r+\big[\sum_{i=1}^{r-1}C_i^r\big]+C_r^r\\ 
                                                                                                       &=\sum_{i=0}^rC_i^r\\
                                                                                                       &=wt({\mathfrak M}{(k, \ell, r)})                                       
\end{align*}
\end{proof}
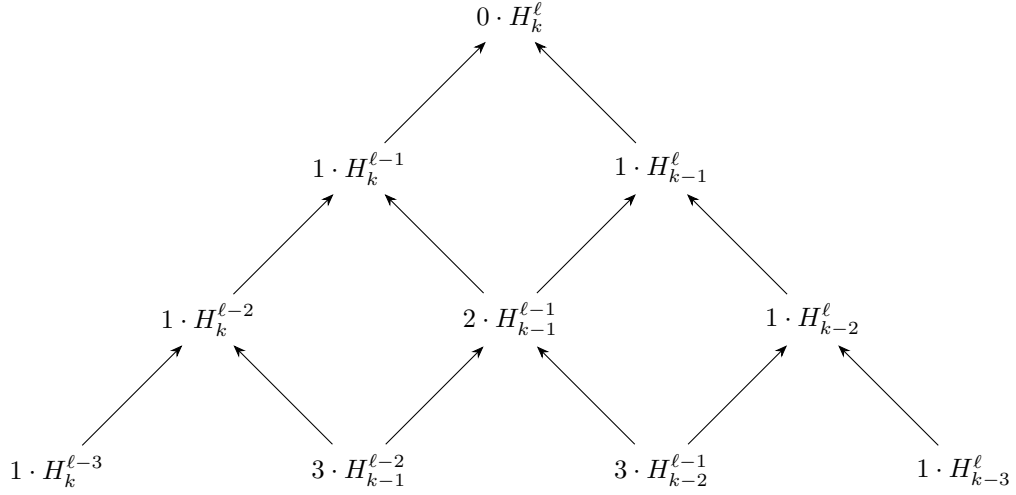
\begin{figure}[H]
\centering
\begin{tikzpicture}[
    >=Stealth,
    node distance=1.8cm and 2.5cm
]
\node (A) at (0,6) {$ 0\cdot H_k^{\ell}$};

\node (B1) at (-2,4) {$1\cdot H_k^{\ell-1}$};
\node (B2) at (2,4) {$1\cdot H_{k-1}^{\ell}$};

\node (C1) at (-4,2) {$1\cdot H_k^{\ell-2}$};
\node (C2) at (0,2) {$2\cdot H_{k-1}^{\ell-1}$};
\node (C3) at (4,2) {$1\cdot H_{k-2}^{\ell}$};

\node (D1) at (-6,0) {$1\cdot H_k^{\ell-3}$};
\node (D2) at (-2,0) {$3\cdot H_{k-1}^{\ell-2}$};
\node (D3) at (2,0) {$3\cdot H_{k-2}^{\ell-1}$};
\node (D4) at (6,0) {$1\cdot  H_{k-3}^{\ell}$};

\draw[->] (B1) -- (A);
\draw[->] (B2) -- (A);

\draw[->] (C1) -- (B1);
\draw[->] (C2) -- (B1);
\draw[->] (C2) -- (B2);
\draw[->] (C3) -- (B2);

\draw[->] (D1) -- (C1);
\draw[->] (D2) -- (C1);
\draw[->] (D2) -- (C2);
\draw[->] (D3) -- (C2);
\draw[->] (D3) -- (C3);
\draw[->] (D4) -- (C3);
\end{tikzpicture}
\caption{The tree rooted at $H_k^{\ell}$ illustrates the distribution of the weights of $H_k^{\ell}(r)$ for $r=0,1,2,3$. Where $c\cdot A_{\alpha}^{\beta}$ means that
$A_{\alpha}^{\beta}$ appears $c$-times in the mosaic ${\mathfrak M}(k,\ell,r)$.}
\label{fig:tree}
\end{figure}
\begin{example}\label{tygffd3u76}
Let $k,\ell\geq 2$ be integers. The codewords of the Fractus code ${\EuScript C}_k^{\ell}$ can be obtained simultaneously by applying the first, second, third, and subsequent fragmentations.
\begin{enumerate}
\item {\bf Initial form or zero-fragmentation:}
Let $X\in {\mathbb F}^{C_{\ell}^{k+\ell-1}}$ a code word. Then  $H_k^{\ell}X=0$.
\item {\bf First fragmentation:} 
Suppose that $X=(X_1, X_2)\in {\mathbb F}^{C_{\ell}^{k+\ell-1}}$ a code word such that  
$X_1\in {\mathbb F}^{C_{\ell-1}^{k+\ell-2}}$ and
$X_2\in {\mathbb F}^{C_{\ell}^{k+\ell-2}}$. 
Then,
$$H_{k}^{\ell}X^T= \left[ \begin{tabular}{c|c}
 $H_k^{\ell-1}$ & $ 0 $ \cr\hline $I_{C_{\ell-1}^{k+\ell-2}}$  & $H_{k-1}^{\ell}$
\end{tabular} \right]
 \begin{bmatrix} X_1^T\\ X_2^T \end{bmatrix}=
\begin{bmatrix} 0^T\\ 0^T \end{bmatrix}$$ 
if and only if
$$\begin{cases}
H_k^{\ell-1}X_1^T&=0^T\\
H_{k-1}^{\ell}X_2^T&=-X_1^T
\end{cases}$$
\item {\bf Second fragmentation:} Suppose that
$X=(X_{11}, X_{12}, X_{21}, X_{22})\in {\mathbb F}^{C_{\ell}^{k+\ell-1}}$ a code word such that  $X_{11}\in {\mathbb F}^{C_{\ell-2}^{k+\ell-3}}$, $X_{12}\in {\mathbb F}^{C_{\ell-1}^{k+\ell-3}}$, $X_{21}\in {\mathbb F}^{C_{\ell-1}^{k+\ell-3}}$, $X_{22}\in {\mathbb F}^{C_{\ell}^{k+\ell-3}}$
\begin{center}
$H_{k}^{\ell}(2)X^T=\left[ \begin{tabular}{c|c|c|c}
 $H_k^{\ell-2}$ & $0$ & $0$ & $0$ 
\cr\hline $I_{C_{\ell-2}^{k+\ell-3}}$ & $H_{k-1}^{\ell-1}$ & $0$ & $0$
 \cr\hline $I_{C_{\ell-2}^{k+\ell-3}}$ & $0$ & $H_{k-1}^{\ell-1}$ & $0$ 
\cr\hline $0$ & $I_{C_{\ell-1}^{k+\ell-3}}$ & $I_{C_{\ell-1}^{k+\ell-3}}$ & $H_{k-2}^{\ell}$
\end{tabular} \right]
\left[\begin{smallmatrix}
X_{11}^T \\
X_{12}^T\\
X_{21}^T \\
X_{22}^T  
\end{smallmatrix}\right]=\left[\begin{smallmatrix}
0^T \\
0^T\\
0^T \\
0^T  
\end{smallmatrix}\right]$
\end{center}
if and only if 
$$\begin{cases}
H_k^{\ell-2}X_{11}^T &= 0^T\\
H_{k-1}^{\ell-1}X_{12}^T &= -X_{11}^T\\
H_{k-1}^{\ell-1}X_{21}^T &= -X_{11}^T\\
H_{k-2}^{\ell}X_{22}^T &= -X_{12}^T-X_{21}^T.
\end{cases}$$
\end{enumerate}
\end{example}
\subsection{Fractus Maximum-Likelihood Decoder}
Suppose that the codeword $X=(X_1,X_2)$ of the Fractus code ${\EuScript C}_k^\ell$ is transmitted in two blocks. Due to channel noise, the received vector is $Y=(Y_1,Y_2)$, with syndrome $S=(S_1,S_2)$. The received syndrome is given by
$$H_{k}^{\ell}Y^T= \left[ \begin{tabular}{c|c}
 $H_k^{\ell-1}$ & $ 0 $ \cr\hline 
 $I_{C_{\ell-1}^{k+\ell-2}}$  & $H_{k-1}^{\ell}$
\end{tabular} \right]
 \begin{bmatrix} Y_1^T\\ Y_2^T \end{bmatrix}=
\begin{bmatrix} S_1^T\\ S_2^T \end{bmatrix},$$ 
  if and only if
 \begin{equation}\label{6treEWW}
\begin{cases}
H_k^{\ell-1}Y_1^T&=S_1^T\\
H_{k-1}^{\ell}Y_2^T&=S_2-Y_1^T.
\end{cases}
\end{equation}
Since $X$ belongs to the Fractus code, $$H_{k}^{\ell}X^T=0^T. $$ Therefore,
$$H_{k}^{\ell}X^T= \left[ \begin{tabular}{c|c}
 $H_k^{\ell-1}$ & $ 0 $ \cr\hline $I_{C_{\ell-1}^{k+\ell-2}}$  & $H_{k-1}^{\ell}$
\end{tabular} \right]
 \begin{bmatrix} X_1^T\\ X_2^T \end{bmatrix}=
\begin{bmatrix} 0^T\\ 0^T \end{bmatrix}$$ 
 lo cual if and only if
\begin{equation}\label{TRFd43}
\begin{cases}
H_k^{\ell-1}X_1^T&=0^T\\
H_{k-1}^{\ell}X_2^T&=-X_1^T
\end{cases}
\end{equation}
For a BSC (Binary Symmetric Channel) with crossover probability $p<1/2$, the maximum-likelihood (ML) decoder seeks
$$\widehat{X}=\arg\max_{X\in{\EuScript C}_k^{\ell}} P(Y\mid X).$$
This is equivalent to
$$\widehat{X}=\arg\min_{X\in{\EuScript C}_k^{\ell}} d_H(X,Y).$$
Since
$$X=(X_1,X_2),\qquad Y=(Y_1,Y_2),$$
we have
$$d_H(X_1,Y_1)+d_H(X_2,Y_2).$$
Therefore, the ML decoding problem can be formulated as
$$\widehat{X}=\arg\min_{(X_1,X_2)}\{d_H(X_1,Y_1)+d_H(X_2,Y_2)\}$$
subject to
$$\begin{cases}
H_k^{\ell-1}X_1^T=0,\\
X_1^T+H_{k-1}^{\ell}X_2^T=0.
\end{cases}$$
Candidates for $X_1$:\\
As $$X_1\in{\EuScript C}_k^{\ell-1} $$
and $$H_k^{\ell-1}Y_1^T=S_1^T ,$$
the first block of errors must satisfy  $$H_k^{\ell-1}E_1^T=S_1^T ,$$
where $$E_1=Y_1-X_1.$$
Therefore, we can find the ML candidates for $X_1$ using the problem
$$\widehat{X}=\arg\min_{X_1\in {\EuScript C}_k^{\ell-1} }d_H(X_1,Y_1)$$
But there is an important subtlety here: it is not necessarily sufficient to choose only the $X_1$ closest to $Y_1$.
Why? Because $X_1$ also determines which $X_2$ can be paired with it.
In other words, $X_1$ determines the problem of finding $X_2$. Once a candidate $X_1$ has been selected, we must find an $X_2$ such that
$H_{k-1}^{\ell}X_2^T=-X_1^T.$\\
 Besides, $$H_{k-1}^{\ell}Y_2^T=S_2^T-Y_1^T.$$ Define $$E_2=Y_2-X_2.$$  Then
$$H_{k-1}^{\ell}E_2^T
=
H_{k-1}^{\ell}Y_2^T
-
H_{k-1}^{\ell}X_2^T.$$  Substituting: 
$$H_{k-1}^{\ell}E_2^T
=
(S_2^T-Y_1^T)+X_1^T.$$
Therefore, $$H_{k-1}^{\ell}E_2^T=S_2^T-Y_1^T+X_1^T.$$
Then
For each candidate $X_1$, we obtain a different syndrome for the second block:
$$S_2(X_1)=S_2-Y_1+X_1.$$
Then we can recursively decode the second block.\\
This naturally leads us to the following procedure.
Fractus Maximum-Likelihood Decoder\\

\begin{center}
\hrulefill\\
{\bf Algorithm 3}\label{GHTAlg6764}\\
\hrulefill\\
\end{center}
 {\bf Input}: $Y=(Y_1,Y_2)$ and $S=(S_1,S_2).$\\
 {\bf Output}: $\widehat {X}\in {\EuScript C}_k^{\ell}$\\
{\bf Step 1:}\\
Calculate
$$S_1=H_k^{\ell-1}Y_1^T.$$
Search for candidates $$X_1\in{\EuScript C}_k^{\ell-1}$$ consistent with the syndrome:
$$H_k^{\ell-1}(Y_1-X_1)^T=S_1.$$
For each candidate  $X_1$, calculate $$S_2(X_1)=S_2-Y_1+X_1.$$
{\bf Step 2:}\\
For all  $X_1$, solve the ML problem corresponding to the second block:
$$\widehat X_2(X_1)
=
\arg\min_{X_2}
d_H(X_2,Y_2)$$ subject to $$H_{k-1}^{\ell}X_2^T=-X_1^T.$$
Equivalently, find

$$E_2(X_1)=\arg\min_E \operatorname{wt}(E)$$
subject to
$$H_{k-1}^{\ell}E^T
=
S_2-Y_1+X_1.$$
Then
\[
\widehat X_2(X_1)=Y_2-E_2(X_1).
\]
{\bf Step 3:}\\
For each candidate $X_1$, we calculate the total cost.
$$D(X_1)
=
d_H(X_1,Y_1)
+
d_H(\widehat X_2(X_1),Y_2).$$
Finally, we select
$$\widehat X_1
=
\arg\min_{X_1}D(X_1)$$ and
$$\widehat {X}=
\left(
\widehat{ X}_1,
\widehat {X}_2(\widehat {X}_1)
\right).
$$
\hrulefill\\
\begin{example}\label{decodingH32GF5}

Consider the matrix
$$H_3^2=
\left[
\begin{array}{ccc|ccc}
1&1&1&0&0&0\cr\hline
1&0&0&1&1&0\\
0&1&0&1&0&1\\
0&0&1&0&1&1
\end{array}
\right].$$
The matrix $H_3^2$ has the block structure
$$H_3^2=
\left[
\begin{array}{c|c}
H_3^1&0\\ \hline
I_3&H_2^2
\end{array}
\right],$$
where
$$H_3^1=
\begin{pmatrix}
1&1&1
\end{pmatrix},
\qquad
H_2^2=
\begin{pmatrix}
1&1&0\\
1&0&1\\
0&1&1
\end{pmatrix}.$$
Let ${\mathbb F}=GF(5)$. We divide a codeword of length $6$ into two blocks of length $3$:
$$X=(X_1,X_2).$$
Consider
$$X_1=
\begin{pmatrix}
1\\
4\\
0
\end{pmatrix},
\qquad
X_2=
\begin{pmatrix}
0\\
4\\
1
\end{pmatrix}.$$
We first verify that $X=(X_1,X_2)$ is a codeword. We have
$$H_3^1X_1^T
=
\begin{pmatrix}
1&1&1
\end{pmatrix}
\begin{pmatrix}
1\\
4\\
0
\end{pmatrix}
=1+4+0=5
\equiv0\pmod5.$$
Moreover,
$$H_2^2X_2^T
=
\begin{pmatrix}
1&1&0\\
1&0&1\\
0&1&1
\end{pmatrix}
\begin{pmatrix}
0\\
4\\
1
\end{pmatrix}
=
\begin{pmatrix}
4\\
1\\
5
\end{pmatrix}
\equiv
\begin{pmatrix}
4\\
1\\
0
\end{pmatrix}
\pmod5.$$
Consequently,
$$X_1^T+H_2^2X_2^T
=
\begin{pmatrix}
1\\
4\\
0
\end{pmatrix}
+
\begin{pmatrix}
4\\
1\\
0
\end{pmatrix}
=
\begin{pmatrix}
5\\
5\\
0
\end{pmatrix}
\equiv
\begin{pmatrix}
0\\
0\\
0
\end{pmatrix}
\pmod5.$$
Thus,
$$X=(1,4,0\,;\,0,4,1)$$
is a codeword.
Now suppose that two errors occur during transmission:
$$E_1=(0,2,0),
\qquad
E_2=(0,0,1).$$
Thus,
$$E=(E_1,E_2)
=(0,2,0\,;\,0,0,1).$$
The received vector is
$$Y=X+E.$$
For the first block,
$$Y_1=(1,4,0)+(0,2,0)
=(1,6,0)
\equiv(1,1,0)\pmod5.$$
For the second block,
$$Y_2=(0,4,1)+(0,0,1)
=(0,4,2).$$
Therefore,
$$Y=(1,1,0\,;\,0,4,2).$$
The syndrome is
$$S=H_3^2Y^T.$$
Using the block structure of $H_3^2$, we write
$$S=
\begin{pmatrix}
S_1\\
S_2
\end{pmatrix},$$
where
$$S_1=H_3^1Y_1^T$$
and
$$S_2=Y_1^T+H_2^2Y_2^T.$$
First,
$$S_1=
\begin{pmatrix}
1&1&1
\end{pmatrix}
\begin{pmatrix}
1\\
1\\
0
\end{pmatrix}
=2.$$
Furthermore,
$$H_2^2Y_2^T
=
\begin{pmatrix}
1&1&0\\
1&0&1\\
0&1&1
\end{pmatrix}
\begin{pmatrix}
0\\
4\\
2
\end{pmatrix}
=
\begin{pmatrix}
4\\
2\\
6
\end{pmatrix}
\equiv
\begin{pmatrix}
4\\
2\\
1
\end{pmatrix}
\pmod5.$$
Therefore,
$$S_2=
\begin{pmatrix}
1\\
1\\
0
\end{pmatrix}
+
\begin{pmatrix}
4\\
2\\
1
\end{pmatrix}
=
\begin{pmatrix}
5\\
3\\
1
\end{pmatrix}
\equiv
\begin{pmatrix}
0\\
3\\
1
\end{pmatrix}
\pmod5.$$
Hence,
$$S=
\begin{pmatrix}
2\\
0\\
3\\
1
\end{pmatrix}.$$
Since $X$ is a codeword,
$$S=H_3^2Y^T
=H_3^2(X+E)^T
=H_3^2E^T.$$
In particular,
$$S_1=H_3^1E_1^T$$
and
$$S_2=E_1^T+H_2^2E_2^T.$$
We now decode the first block.
The code associated with $H_3^1$ is
$${\mathcal} C_3^1
=
\left\{
(a,b,c)\in GF(5)^3:
a+b+c=0
\right\}.$$
The received block is
$$Y_1=(1,1,0).$$
We seek a codeword minimizing the Hamming distance:
$$\widehat X_1
=
\underset{Z\in{\mathcal C}_3^1}{\operatorname{arg\,min}}
d_H(Y_1,Z).$$
Since
$$1+1+0=2\not\equiv0\pmod5,$$
we have
$$Y_1\notin{\mathcal C}_3^1.$$
Changing the second coordinate from $1$ to $4$ gives
$$
1+4+0=5\equiv0\pmod5.
$$
Therefore,
$$\widehat X_1=(1,4,0).$$
The corresponding error estimate is
$$\widehat E_1
=
Y_1-\widehat X_1.$$
Thus,
$$\widehat E_1
=
(1,1,0)-(1,4,0)
=
(0,-3,0)
\equiv
(0,2,0)
\pmod5.$$
Hence,
$$\widehat E_1=(0,2,0)=E_1.$$
We have therefore correctly recovered the first error.
We now use the block structure of the Fractus matrix to decode the second block.
Since every codeword satisfies
$$X_1^T+H_2^2X_2^T=0,$$
we have
$$H_2^2X_2^T=-X_1^T.$$
On the other hand,
$$S_2=Y_1^T+H_2^2Y_2^T.$$
Since
$$Y_2=X_2+E_2,$$
we obtain
$$H_2^2Y_2^T
=
H_2^2X_2^T+H_2^2E_2^T.$$
Consequently,
$$S_2
=
Y_1^T-X_1^T+H_2^2E_2^T,$$
and hence
$$H_2^2E_2^T
=
S_2-Y_1^T+X_1^T.$$
Since $\widehat X_1=X_1$, we can compute
$$H_2^2E_2^T
=
S_2-Y_1^T+\widehat X_1^T.$$
In our example,
$$S_2=
\begin{pmatrix}
0\\
3\\
1
\end{pmatrix},
\qquad
Y_1=
\begin{pmatrix}
1\\
1\\
0
\end{pmatrix},
\qquad
\widehat X_1=
\begin{pmatrix}
1\\
4\\
0
\end{pmatrix}.$$
Therefore,
$$H_2^2E_2^T
=
\begin{pmatrix}
0\\
3\\
1
\end{pmatrix}
-
\begin{pmatrix}
1\\
1\\
0
\end{pmatrix}
+
\begin{pmatrix}
1\\
4\\
0
\end{pmatrix}
=
\begin{pmatrix}
0\\
6\\
1
\end{pmatrix}
\equiv
\begin{pmatrix}
0\\
1\\
1
\end{pmatrix}
\pmod5.$$
Now observe that
$$H_2^2
\begin{pmatrix}
0\\
0\\
1
\end{pmatrix}
=
\begin{pmatrix}
0\\
1\\
1
\end{pmatrix}.$$
Therefore,
$$\widehat E_2=(0,0,1).$$
This agrees with the actual error introduced during transmission.
Finally,
$$\widehat X_2
=
Y_2-\widehat E_2
=
(0,4,2)-(0,0,1)
=
(0,4,1).$$
Thus,
$$\widehat X_2=(0,4,1)=X_2.$$
Consequently, the decoded codeword is
$$\widehat {X}
=
(\widehat {X}_1,\widehat {X}_2)
=
(1,4,0\,;\,0,4,1).$$
Hence,
$$\widehat {X}=X.$$

This example illustrates how the recursive block structure of the Fractus matrix can be exploited in the decoding process. First, the block associated with $H_3^1$ is decoded. The recovered estimate $\widehat X_1$ is then used to construct a corrected syndrome for the second block. Finally, the second error is recovered by solving the resulting minimum-weight syndrome equation involving $H_2^2$.
In particular, the decoding procedure exploits the recursive structure
$$H_3^2=
\left[
\begin{array}{c|c}
H_3^1&0\cr\hline
I_3&H_2^2
\end{array}
\right],$$
rather than treating $H_3^2$ as an unstructured parity-check matrix.
\end{example}

\section{Lattices Associated with Matrices and Codes}
\begin{defn}\label{t543ede}
Let ${\mathbb N}_0=\{0,1,2,\ldots\}$. We define the set of fractus matrices as
$${\mathscr FM}=\{ H_k^{\ell} : (k, \ell)\in {\mathbb N}_0\times {\mathbb N}_0\} $$
 The order in the set  ${\mathscr FM}$ is given by $H_k^{\ell}\leq H_{k^{\prime}}^{\ell^{\prime}}$ if and only if $k\leq k^{\prime}$ and 
$\ell\leq \ell^{\prime}.$\\
For each $H_k^{\ell}$ and  $H_{k^{\prime}}^{\ell^{\prime}}$ in ${\mathscr FM}$. We define 
$$H_k^{\ell}\vee H_{k^{\prime}}^{\ell^{\prime}}=H_{\max\{k, k^{\prime}\}}^{\max\{\ell, \ell^{\prime}\}}$$ and
$$H_k^{\ell}\wedge H_{k^{\prime}}^{\ell^{\prime}}=H_{\min\{k, k^{\prime}\}}^{\min\{\ell, \ell^{\prime}\}}$$  
\end{defn}
%

%
\begin{prop}\label{join-meet}
$({\mathscr FM}, \vee, \wedge )$ is a distributive lattice with a minimum element and no maximum element.
\end{prop}
\begin{proof}
First  $({\mathscr FM}, \leq)$ is a partially ordered set. $H_0^0$ is the minimum element, since $H_0^0\wedge H_k^{\ell}=H_0^0$ and  $H_0^0\vee H_k^{\ell}=H_k^{\ell}$.\\
For each $H_k^{\ell}$ and  $H_{k^{\prime}}^{\ell^{\prime}}$ in ${\mathscr FM}$ , it is easy to see that
 $H_k^{\ell}\wedge H_{k^{\prime}}^{\ell^{\prime}}=H_{\min\{k, k^{\prime}\}}^{\min\{\ell, \ell^{\prime}\}}$ is a meet and
 $H_k^{\ell}\vee H_{k^{\prime}}^{\ell^{\prime}}=H_{\max\{k, k^{\prime}\}}^{\max\{\ell, \ell^{\prime}\}}$ is a join.  
 Proving commutativity, associativity, and idempotency is left as an easy exercise for the reader. Here we prove distributive laws and absorption.
$$\begin{aligned}
H_{k_1}^{\ell_1}\wedge(H_{k_2}^{\ell_2}\vee H_{k_3}^{\ell_3})
&=
H_{\max\{\min(k_1,k_2),\min(k_1,k_3)\}}
  ^{\max\{\min(\ell_1,\ell_2),\min(\ell_1,\ell_3)\}}\\
&=
(H_{k_1}^{\ell_1}\wedge H_{k_2}^{\ell_2})
\vee
(H_{k_1}^{\ell_1}\wedge H_{k_3}^{\ell_3}).
\end{aligned}$$

Similarly, identity 
$$H_{k_1}^{\ell_1}\vee(H_{k_2}^{\ell_2}\wedge H_{k_3}^{\ell_3})=(H_{k_1}^{\ell_1}\vee H_{k_2}^{\ell_2})\wedge (H_{k_1}^{\ell_1}\vee H_{k_3}^{\ell_3})$$
can be proven.\\
The absorption law,
$$H_k^\ell\vee(H_k^\ell\wedge H_{k'}^{\ell'})
=
H_{\max\{k,\min(k,k')\}}^{\max\{\ell,\min(\ell,\ell')\}}
=
H_k^\ell.$$
In the same way we demonstrated $H_k^{\ell}\wedge (H_k^{\ell}\vee H_{k^{\prime}}^{\ell^{\prime}})=H_k^{\ell}$\\
In $({\mathscr FM}, \vee, \wedge)$, there is no maximum element, since for every $H_k^\ell \in {\mathscr FM}$,
$$H_k^\ell < H_{k+1}^{\ell+1}.$$
\end{proof}
So far, we have shown that $({\mathscr FM}, \vee, \wedge)$ is a lattice.
We now introduce the following notation. For $n\in\mathbb{N}$, define
\begin{equation}\label{Not54rde}
nH_k^\ell=
\underbrace{H_k^\ell\vee H_k^\ell\vee\cdots\vee H_k^\ell}_{n\text{ times}}.
\end{equation}
Since $\vee$ is idempotent, it follows that
\begin{equation}\label{TTy4rde67}
nH_k^\ell=H_k^\ell.
\end{equation}
\begin{defn}
For $k,\ell\geq 2$ and $r\in{\mathbb N}_0$ such that
$0\leq r\leq \min\{k-1,\ell-1\}$, the \textit{$r$-expansion} of
$H_k^\ell$, denoted by $\operatorname{Exp}_r(H_k^\ell)$, is
recursively defined by
\begin{align*}
\operatorname{Exp}_0(H_k^\ell)&=H_k^\ell,\\
\operatorname{Exp}_r(H_k^\ell)
&=\operatorname{Exp}_{r-1}(H_k^{\ell-1})
  \vee \operatorname{Exp}_{r-1}(H_{k-1}^\ell),
  \qquad r\geq 1.
\end{align*}
\end{defn}

\begin{example}\label{GYTR453}
We compute the $3$-expansion of $A_8^7$. By the recursive definition,
\begin{align*}
\operatorname{Exp}_3(H_8^7)
&=\operatorname{Exp}_2(H_8^6)
  \vee \operatorname{Exp}_2(H_7^7)\\
&=\bigl(\operatorname{Exp}_1(H_8^5)
  \vee \operatorname{Exp}_1(H_7^6)\bigr)
  \vee
  \bigl(\operatorname{Exp}_1(H_7^6)
  \vee \operatorname{Exp}_1(H_6^7)\bigr)\\
&=\bigl((H_8^4\vee H_7^5)
  \vee(H_7^5\vee H_6^6)\bigr)
  \vee
  \bigl((H_7^5\vee H_6^6)
  \vee(H_6^6\vee H_5^7)\bigr).
\end{align*}
Using associativity, commutativity, and idempotence of $\vee$, we obtain
$$\operatorname{Exp}_3(H_8^7)
=H_8^4\vee H_7^5\vee H_6^6\vee H_5^7.$$
\end{example}
We now derive an explicit formula for the $r$-expansion in terms of the standard notation for binomial coefficients.
\begin{prop}\label{dfrim43443}
Let $k,\ell \geq 2$ and $0 \leq r \leq \min\{k-1,\ell-1\}$. Then the $r$-expansion of $H_k^\ell$ is given by
$$\operatorname{Exp}_r(H_k^\ell)
=\bigvee_{i=0}^{r}
\binom{r}{i}\, 
H_{k-i}^{\ell-(r-i)}.$$
Moreover, the total number of terms, counted with multiplicity, is
$$\sum_{i=0}^{r}\binom{r}{i}=2^r.$$
\end{prop}
\begin{proof}
The proof proceeds by induction on $r$. The case $r=0$ is immediate. For $r=1$, the definition gives
$$\operatorname{Exp}_1(H_k^\ell)
=
H_k^{\ell-1}\vee H_{k-1}^\ell.$$
For $r=2$, we obtain
\begin{align*}
\operatorname{Exp}_2(H_k^\ell)
&=
\operatorname{Exp}_1(H_k^{\ell-1})
\vee
\operatorname{Exp}_1(H_{k-1}^\ell)\\
&=
\bigl[H_k^{\ell-2}\vee H_{k-1}^{\ell-1}\bigr]
\vee
\bigl[H_{k-1}^{\ell-1}\vee H_{k-2}^\ell\bigr]\\
&=
\binom{2}{0}H_k^{\ell-2}
\vee
\binom{2}{1}H_{k-1}^{\ell-1}
\vee
\binom{2}{2}H_{k-2}^\ell.
\end{align*}
Moreover,
$$\binom{2}{0}+\binom{2}{1}+\binom{2}{2}=2^2.$$

Now, assume as the induction hypothesis that
\begin{equation}\label{YTrf543}
\begin{aligned}
\operatorname{Exp}_{r-1}(H_k^{\ell-1})
&=
H_k^{\ell-r}
\vee
\bigvee_{i=1}^{r-1}
\binom{r-1}{i}H_{k-i}^{\ell-(r-i)},
\quad
\sum_{i=0}^{r-1}\binom{r-1}{i}=2^{r-1},\\
\operatorname{Exp}_{r-1}(H_{k-1}^\ell)
&=
\bigvee_{i=1}^{r-1}
\binom{r-1}{i-1}H_{k-i}^{\ell-(r-i)}
\vee
H_{k-r}^\ell,
\quad
\sum_{i=1}^{r}\binom{r-1}{i-1}=2^{r-1}.
\end{aligned}
\end{equation}
By the definition of the expansion operator and the induction hypothesis in \eqref{YTrf543}, we obtain
\begin{align*}
\operatorname{Exp}_r(H_k^\ell)
&=
\operatorname{Exp}_{r-1}(H_k^{\ell-1})
\vee
\operatorname{Exp}_{r-1}(H_{k-1}^\ell)\\
&=
\left[
H_k^{\ell-r}
\vee
\bigvee_{i=1}^{r-1}
\binom{r-1}{i}H_{k-i}^{\ell-(r-i)}
\right]
\vee
\left[
\bigvee_{i=1}^{r-1}
\binom{r-1}{i-1}H_{k-i}^{\ell-(r-i)}
\vee
H_{k-r}^\ell
\right]\\
&=
H_k^{\ell-r}
\vee
\bigvee_{i=1}^{r-1}
\left[
\binom{r-1}{i}
+
\binom{r-1}{i-1}
\right]
H_{k-i}^{\ell-(r-i)}
\vee
H_{k-r}^\ell\\
&=
H_k^{\ell-r}
\vee
\bigvee_{i=1}^{r-1}
\binom{r}{i}H_{k-i}^{\ell-(r-i)}
\vee
H_{k-r}^\ell\\
&=
\bigvee_{i=0}^{r}
\binom{r}{i}H_{k-i}^{\ell-(r-i)}.
\end{align*}

Moreover, the total number of terms, counted with multiplicity, is
\begin{align*}
\sum_{i=0}^{r-1}\binom{r-1}{i}
+
\sum_{i=1}^{r}\binom{r-1}{i-1}
&=
2^{r-1}+2^{r-1}\\
&=
2^r.
\end{align*}

Therefore,
$$
\operatorname{Exp}_r(H_k^\ell)
=
\bigvee_{i=0}^{r}
\binom{r}{i}H_{k-i}^{\ell-(r-i)},
$$
and
$$
\sum_{i=0}^{r}\binom{r}{i}=2^r.
$$
\end{proof}
%



\begin{defn}\label{gtfre33321}
Let $H_k^\ell \in \mathscr{FM}$ be arbitrary. Define
\[
\Omega_{H_k^\ell}
=
\left\{
H_\alpha^\beta \in \mathscr{FM}
\;\middle|\;
H_\alpha^\beta \leq H_k^\ell
\right\}.
\]
\end{defn}

\begin{lem}\label{yytred345}
\begin{enumerate}
    \item $H_\alpha^\beta$ is a submatrix of $H_k^\ell$ if and only if
    $H_\alpha^\beta \in \Omega_{H_k^\ell}$.
    
    \item The set $\Omega_{H_k^\ell}$ is bounded with respect to the order
    $\leq$.
\end{enumerate}
\end{lem}

\begin{proof}
For (1), $H_\alpha^\beta$ is an $(\alpha,\beta)$-regular submatrix of the
$(k,\ell)$-regular matrix $H_k^\ell$ if and only if
$\alpha \leq k$ and $\beta \leq \ell$. Equivalently,
\[
H_\alpha^\beta \leq H_k^\ell,
\]
which is equivalent to $H_\alpha^\beta \in \Omega_{H_k^\ell}$.

For (2), $H_0^0 = 0$ is a lower bound for $\Omega_{H_k^\ell}$, while
$H_k^\ell$ is an upper bound. Therefore,
$\Omega_{H_k^\ell}$ is bounded.
\end{proof}
\begin{rem}
We have 
$$H_k^{\ell}=  \left[ \begin{tabular}{c|c}
 $H_k^{\ell-1}$ & $ 0 $ \cr\hline $I_{C_{\ell-1}^{k+\ell -2}}$  & $H_{k-1}^{\ell}$
\end{tabular} \right] $$
Hence  $H_{k-1}^{\ell-1}\unlhd  H_k^{\ell-1}\unlhd  H_k^{\ell}$ and  $H_{k-1}^{\ell-1}\unlhd  H_{k-1}^{\ell}\unlhd  H_k^{\ell}$ \\
%
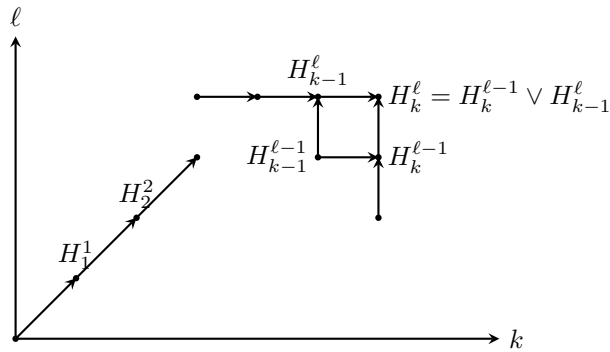
\begin{figure}[htbp]
\centering
\begin{tikzpicture}[>=stealth,scale=0.8]
\draw[->,thick] (0,0) -- (8,0) node[right] {$k$};
\draw[->,thick] (0,0) -- (0,5) node[above] {$\ell$};

\coordinate (A) at (5,4);
\coordinate (B) at (6,4);
\coordinate (C) at (6,3);
\coordinate (D) at (5,3);
\coordinate (D1) at (3,4);
\coordinate (E) at (6,2);
\coordinate (F) at (4,4);
\coordinate (X) at (0,0);
\coordinate (Y) at (1,1);
\coordinate (Z) at (2,2);
\coordinate (M) at (3,3);

\fill (A) circle (1.5pt);
\fill (B) circle (1.5pt);
\fill (C) circle (1.5pt);
\fill (D) circle (1.5pt);
\fill (D1) circle (1.5pt);
\fill (E) circle (1.5pt);
\fill (F) circle (1.5pt);
\fill (X) circle (1.5pt);
\fill (Y) circle (1.5pt);
\fill (Z) circle (1.5pt);
\fill (M) circle (1.5pt);

\draw[->,thick] (A) -- (B);
\draw[->,thick] (C) -- (B);
\draw[->,thick] (D) -- (A);
\draw[->,thick] (D) -- (C);
\draw[->,thick] (E) -- (C);
\draw[->,thick] (F) -- (A);
\draw[->,thick] (X) -- (Y);
\draw[->,thick] (Y) -- (Z);
\draw[->,thick] (Z) -- (M);
\draw[->,thick] (D1) -- (F);

\node[above] at (A) {$H_{k-1}^{\ell}$};
\node[right] at (B) {$H_k^\ell = H_k^{\ell-1} \vee H_{k-1}^\ell$};
\node[right] at (C) {$H_k^{\ell-1}$};
\node[left] at (D) {$H_{k-1}^{\ell-1}$};
\node[above] at (Y) {$H_1^1$};
\node[above] at (Z) {$H_2^2$};
\end{tikzpicture}
\caption{The figure represents $\Omega_{H_k^\ell}$, where the arrows indicate the
submatrix relation.}
\label{fig:local-lattice-structure}
\end{figure}%
\end{rem}
%
\begin{prop}\label{id3al8766}
The set $\Omega_{H_k^\ell}$ is the principal ideal generated by
$H_k^\ell$. Moreover,
$$\lvert \Omega_{H_k^\ell} \rvert = (k+1)(\ell+1).$$
\end{prop}
\begin{proof}
By definition,
$$\Omega_{H_k^\ell}
=
\left\{
H_\alpha^\beta \in \mathscr{FM}
\;\middle|\;
H_\alpha^\beta \leq H_k^\ell
\right\}$$
Thus, $\Omega_{H_k^\ell}$ is the principal ideal generated by
$H_k^\ell$.

Indeed, if $H_{\alpha_1}^{\beta_1},H_{\alpha_2}^{\beta_2}
\in \Omega_{H_k^\ell}$, then
$$H_{\alpha_1}^{\beta_1}\vee H_{\alpha_2}^{\beta_2}
=
H_{\max\{\alpha_1,\alpha_2\}}^{\max\{\beta_1,\beta_2\}}
\leq H_k^\ell.$$

Therefore,
$$H_{\alpha_1}^{\beta_1}\vee H_{\alpha_2}^{\beta_2}
\in \Omega_{H_k^\ell}.$$  Also, if
$H_{\alpha'}^{\beta'}\leq H_\alpha^\beta$ and
$H_\alpha^\beta\in\Omega_{H_k^\ell}$, then
$H_{\alpha'}^{\beta'}\leq H_k^\ell$.\\
 Hence,
$$H_{\alpha'}^{\beta'}\in\Omega_{H_k^\ell}$$.

Since $H_\alpha^\beta\leq H_k^\ell$ if and only if
$\alpha\leq k$ and $\beta\leq\ell$, we have
\[
\Omega_{H_k^\ell}
=
\left\{
H_\alpha^\beta
\;\middle|\;
(\alpha,\beta)\in
\{0,\ldots,k\}\times\{0,\ldots,\ell\}
\right\}.
\]
Thus,
\[
\lvert\Omega_{H_k^\ell}\rvert=(k+1)(\ell+1).
\]
\end{proof}%
\begin{lem}
Let $k,\ell\geq 0$.
\begin{enumerate}
\item If $k=k_0$ is fixed, then
$$\Omega_{H_{k_0}^0}\subseteq \Omega_{H_{k_0}^1}\subseteq
\Omega_{H_{k_0}^2}\subseteq\cdots\subseteq
\Omega_{H_{k_0}^{\ell-1}}\subseteq
\Omega_{H_{k_0}^{\ell}}\subseteq\cdots.$$
\item If $\ell=\ell_0$ is fixed, then
$$\Omega_{H_0^{\ell_0}}\subseteq \Omega_{H_1^{\ell_0}}\subseteq
\Omega_{H_2^{\ell_0}}\subseteq\cdots\subseteq
\Omega_{H_{k-1}^{\ell_0}}\subseteq
\Omega_{H_k^{\ell_0}}\subseteq\cdots.$$
\end{enumerate}
\end{lem}
\begin{proof}
By Lemma \ref{yytred345}, the relations
\[
H_{k_0}^j\unlhd H_{k_0}^{j+1}
\quad\text{and}\quad
H_j^{\ell_0}\unlhd H_{j+1}^{\ell_0},
\]
for every $j\geq 0$, imply the corresponding inclusions between
the sets $\Omega_{H_{k_0}^j}$ and $\Omega_{H_j^{\ell_0}}$,
respectively.
\end{proof}
 \begin{lem}\label{FM5643}
Let $k,\ell\geq 1$ be integers, and consider Definition \ref{fed789865r}. Then $(H_k^{\ell})^F=H_{\ell}^k$
\end{lem}
\begin{proof} 
The proof is by induction over $k\geq 1$ and $\ell\geq 1$, so clealy $$(H_k^1)^F=(1,1, \ldots, 1)^F= \begin{bmatrix} 1\\ 1\\ \vdots\\ 1 \end{bmatrix}= H_1^k$$
Analogaus for  $(H_1^{\ell})^F=H_{\ell}^1$.
Hence, for all $k^{\prime}\geq 2$ and for each $\ell^{\prime}\geq 2$ we can assume $(H_{{k^{\prime}}}^2)^F=H_2^{k^{\prime}}$ and
$(H_2^{\ell^{\prime}})^F=H_{\ell^{\prime}}^2$. Let $k\geq k^{\prime}\geq 2$ y $\ell\geq \ell^{\prime}\geq 2$. Then 
$$(H_k^2)^F= \left[ \begin{tabular}{c|c}
 $(H_{k-1}^2)^F$ &  0  \cr\hline $I_k$  & $(H_k^1)^F$
\end{tabular} \right]
=\left[ \begin{tabular}{c|c}
 $H_2^{k-1}$ &  0  \cr\hline $I_k$  & $H_1^k$
\end{tabular} \right]=H_2^k$$
Analogously, we can show that  $(H_2^{\ell})^F=H_{\ell}^2$.\\
Now suppose, as the induction hypothesis, that for every $2\leq k^{\prime} \leq k$ 
and every $2\leq \ell^{\prime} \leq \ell$ we have $(H_k^{ \ell^{\prime}})^F=H_{ \ell^{\prime}}^k$
and $(H_{\ell}^{k^{\prime}})^F=H_{ k^{\prime}}^{\ell}$. As a consequence of this fact, we obtain
$$(H_k^{\ell})^F= \left[ \begin{tabular}{c|c}
 $(H_{k-1}^{\ell})^F$ & $ 0 $ \cr\hline $I_{C_{\ell-1}^{k+\ell-2}}$  & $(H_k^{\ell-1})^F$
\end{tabular} \right]
=\left[ \begin{tabular}{c|c}
 $H_{\ell}^{k-1}$ & $ 0 $ \cr\hline $I_{C_{k-1}^{k+\ell-2}}$  & $H_{\ell-1}^k$
\end{tabular} \right]=H_{\ell}^k$$
\end{proof}
\begin{defn}\label{Tgyh543}
Let $\mathscr{FM}=(\mathscr{FM},\vee,\wedge)$ be the associative algebra.
By Lemma \ref{FM5643}, the map
\begin{equation}\label{hhgttr432}
\begin{aligned}
()^F:\mathscr{FM} &\longrightarrow \mathscr{FM},\\
H_k^\ell &\longmapsto H_\ell^k.
\end{aligned}
\end{equation}
is well defined.
\end{defn}
\begin{lem}\label{7uy6TTT}
With the notation of Definition \ref{fed789865r}, the map $F$ defined in Definition \ref{Tgyh543} satisfies the following properties:
\begin{description}
\item [ a)] $((H_k^{\ell})^F)^F=H_k^{\ell}$ \\
\item [ b)] If $H_k^{\ell}\leq H_{k^{\prime}}^{\ell^{\prime}}$ then $(H_k^{\ell})^F\leq (H_{k^{\prime}}^{\ell^{\prime}})^F$ \\
\item [ c)] $(H_k^{\ell}\vee H_{k^{\prime}}^{\ell^{\prime}})^F=(H_k^{\ell})^F\vee (H_{k^{\prime}}^{\ell^{\prime}})^F$ \\
\item [ d)] $(H_k^{\ell}\wedge H_{k^{\prime}}^{\ell^{\prime}})^F=(H_k^{\ell})^F\wedge (H_{k^{\prime}}^{\ell^{\prime}})^F$ \\
\item [ e)] $(H_k^{\ell}\vee[H_{k^{\prime}}^{\ell^{\prime}} \wedge H_{k^{\prime\prime}}^{\ell^{\prime\prime}}])^F=(H_k^{\ell})^F\vee[(H_{k^{\prime}}^{\ell^{\prime}})^F \wedge (H_{k^{\prime\prime}}^{\ell^{\prime\prime}})^F]$\\
\item [ f)] $(H_k^{\ell}\wedge[H_{k^{\prime}}^{\ell^{\prime}}\vee H_{k^{\prime\prime}}^{\ell^{\prime\prime}}])^F=(H_k^{\ell})^F\wedge[(H_{k^{\prime}}^{\ell^{\prime}})^F\vee
(H_{k^{\prime\prime}}^{\ell^{\prime\prime}})^F]$\\
\item [ g)] If $H_k^{\ell}=\bigvee_{i=0}^r\binom{r}{i} H_{k-i}^{\ell-(r-i)}$, then 
$$(H_k^{\ell})^F=\bigvee_{i=0}^r\binom{r}{i} H_{\ell-(r-i)}^{k-i}$$
\end{description}
\end{lem}
\begin{proof}
Items (a) and (b)  are easily verifiable.\\
For item (c) we have
 \begin{align*}
     (H_k^{\ell}\vee H_{k^{\prime}}^{\ell^{\prime}})^F&=(H_{\max\{k,  k^{\prime}\}}^{\max\{\ell,  \ell^{\prime}\}})^F\\
                                                                               &=H_{\max\{\ell,  \ell^{\prime}\}}^{\max\{k,  k^{\prime}\}}\\
                                                                               &=H_{\ell}^k\vee H_{\ell^{\prime}}^{k^{\prime}}\\
                                                                               &=(H_k^{\ell})^F\vee (H_{k^{\prime}}^{\ell^{\prime}})^F\\
    \end{align*}
 The proof of item (d)  is the same as in (c).\\
The proof of (e) and (f) is easy and the proof of item (g) following item (c).
\end{proof}
\begin{lem}\label{GTr43deF}
The $(k,\ell)$-Fractus matrix $H_k^{\ell}$ is a square matrix if and only if $k=\ell$.
\end{lem}
\begin{proof}
The $(k,\ell)$-Fractus matrix $H_k^{\ell}$ has dimensions $C_{\ell-1}^{k+\ell-1}\times C_{\ell}^{k+\ell-1}$. Hence,
$$C_{\ell-1}^{k+\ell-1}=C_{\ell}^{k+\ell-1},$$
which holds if and only if $$k=\ell$$.
\end{proof}
\begin{prop}\label{gro765gf}
The map $()^F$ is an automorphism of the lattice algebra
$\mathscr{FM}$. Moreover, its set of fixed points consists precisely of the family of square matrices.
$$\{H_k^k : k\geq 0\}.$$
\end{prop}
\begin{proof}
By Lemma \ref{7uy6TTT}(a), the map $()^F$ is an involution. Hence, it is
bijective. Furthermore, Lemma \ref{7uy6TTT}(c) and (d) show that
$$(A\vee B)^F=A^F\vee B^F
\quad\text{and}\quad
(A\wedge B)^F=A^F\wedge B^F$$
for all $A,B\in\mathscr{FM}$. Therefore, $()^F$ is an automorphism of
the lattice algebra $\mathscr{FM}$.
Now, $H_k^\ell$ is fixed by $F$ if and only if
$$(H_k^\ell)^F=H_k^\ell,$$
which is equivalent to
$$H_\ell^k=H_k^\ell.$$
Since $H_k^\ell=H_\ell^k$ if and only if $k=\ell$, the fixed points of
$()^F$ are precisely the elements $H_k^k$, with $k\geq 0$.
\end{proof}
\subsection{Lattice ${\mathscr H}$} Next, we give a lattice induced by  ${\mathscr FM}$. 
\begin{defn}\label{2t543ede}
For all $k, \ell\geq 2$, let ${\EuScript C}_k^{\ell}$ denote the $(\ell, k)$-fractus code  (see, Definition \ref{Akl4327}). We then define
$${\mathscr H}=\{ {\EuScript C}_k^{\ell}: (k, \ell)\in {\mathbb N}_0\times {\mathbb N}_0\}$$
\end{defn}
Recall that $H_{k^{\prime}}^{\ell^{\prime}}\unlhd H_k^{\ell}$ means that  $H_{k^{\prime}}^{\ell^{\prime}}$ is a submatrix $ H_k^{\ell}$.  Then 
$${\EuScript C}_k^{\ell}=\ker H_k^{\ell} \subseteq \ker H_{k^{\prime}}^{\ell^{\prime}}={\EuScript C}_{k^{\prime}}^{\ell^{\prime}}$$
\begin{defn}\label{6tre453}
We define the parcial order in ${\mathscr H}$
\begin{center}
${\EuScript C}_k^{\ell} \leq{\EuScript C}_{k^{\prime}}^{\ell^{\prime}}$ if and only if  $H_{k^{\prime}}^{\ell^{\prime}}\leq H_k^{\ell}$
\end{center}
\end{defn}
\begin{rem}\label{67trde111}
${\EuScript C}_k^{\ell} \leq{\EuScript C}_{k^{\prime}}^{\ell^{\prime}}$ if and only if  $k^{\prime} \leq k$, $\ell^{\prime}\leq \ell$ and  ${\EuScript C}_k^{\ell} \subseteq {\EuScript C}_{k^{\prime}}^{\ell^{\prime}}$\\
\end{rem}
It is easy to see from Remark \ref{67trde111} that $(\mathscr{H}, \leq)$ is a partially ordered set.\\
For each ${\EuScript C}_k^{\ell}$ and  ${\EuScript C}_{k^{\prime}}^{\ell^{\prime}}$ in ${\mathscr H}$ we define\\
meet; $${\EuScript C}_k^{\ell}\wedge {\EuScript C}_{k^{\prime}}^{\ell^{\prime}}={\EuScript C}_{\max\{k, k^{\prime}\}}^{\max\{\ell, \ell^{\prime}\}}$$ and
join; $${\EuScript C}_k^{\ell}\vee {\EuScript C}_{k^{\prime}}^{\ell^{\prime}}={\EuScript C}_{\min\{k, k^{\prime}\}}^{\min\{\ell, \ell^{\prime}\}}$$. 
\begin{prop}
$({\mathscr H},  \vee, \wedge)$ is a lattice and an associative algebra.
\end{prop}
\begin{proof}
Clearly ${\EuScript C}_k^{\ell}\wedge {\EuScript C}_{k^{\prime}}^{\ell^{\prime}}\leq {\EuScript C}_k^{\ell}$ if and only if  
${\EuScript C}_{\max\{k, k^{\prime}\}}^{\max\{\ell, \ell^{\prime}\}}\leq {\EuScript C}_k^{\ell}$, similarly, it is demonstrated that 
${\EuScript C}_k^{\ell}\wedge {\EuScript C}_{k^{\prime}}^{\ell^{\prime}}\leq {\EuScript C}_{k^{\prime}}^{\ell^{\prime}}$       .\\
Now ${\EuScript C}_k^{\ell}\leq {\EuScript C}_k^{\ell}\vee {\EuScript C}_{k^{\prime}}^{\ell^{\prime}}$ if and only if ${\EuScript C}_k^{\ell}\leq {\EuScript C}_{\min\{k, k^{\prime}\}}^{\min\{\ell, \ell^{\prime}\}}$  
similarly, it is demonstrated that  
${\EuScript C}_{k^{\prime}}^{\ell^{\prime}}\leq {\EuScript C}_k^{\ell}\vee {\EuScript C}_{k^{\prime}}^{\ell^{\prime}}$. This shows that $({\mathscr H}, \vee, \wedge)$ is a lattice.
\end{proof}
\begin{defn}\label{554rde32}
Let ${\EuScript C}_k^{\ell}\in {\mathscr H}$ be any element of ${\EuScript H}$. Define
\begin{equation}\label{jhuy654}
\Sigma_{{\EuScript C}_k^{\ell}}=\{ {\EuScript C}_{\alpha}^{\beta} \in {\mathscr H}:  {\EuScript C}_k^{\ell}\leq {\EuScript C}_{\alpha}^{\beta} \}
\end{equation}
\end{defn}
There is a bijective correspondence:
\begin{equation}\label{treeew321}
\begin{aligned}
\rho: {\mathscr FM} &\longrightarrow {\mathscr H}\\
                           H_k^{\ell}& \longmapsto {\EuScript C}_k^{\ell}
\end{aligned}
\end{equation}
Clearly, the inverse image of ${\EuScript C}_k^{\ell}$ under $\rho$ is
$\rho^{-1}({\EuScript C}_k^{\ell})=H_k^{\ell}.$
\begin{lem}\label{Ttre4321}
Consider the function $\rho$ defined in (\ref{treeew321}). Then the following properties hold:
\begin{description}
 \item[1)] If $H_{k^{\prime}}^{\ell^{\prime}}\leq H_k^{\ell}$ then  ${\EuScript C}_k^{\ell} \leq {\EuScript C}_{k^{\prime}}^{\ell^{\prime}}$\\
\item[2)] $\rho(H_k^{\ell}\wedge  H_{k^{\prime}}^{\ell^{\prime}})= {\EuScript C}_k^{\ell}\vee {\EuScript C}_{k^{\prime}}^{\ell^{\prime}}$\\
\item[3)] $\rho(H_k^{\ell}\vee  H_{k^{\prime}}^{\ell^{\prime}})=  {\EuScript C}_k^{\ell}\wedge {\EuScript C}_{k^{\prime}}^{\ell^{\prime}}$\\
\item[4)] $\rho(H_{k}^{\ell}\wedge(H_{k^{\prime}}^{\ell^{\prime}}\vee H_{k^{\prime\prime}}^{\ell^{\prime\prime}}))={\EuScript C}_k^{\ell}\vee ({\EuScript C}_{k^{\prime}}^{\ell^{\prime}}\wedge{\EuScript C}_{k^{\prime\prime}}^{\ell^{\prime\prime}})$\\
\item[5)] $\rho(H_{k}^{\ell}\vee(H_{k^{\prime}}^{\ell^{\prime}}\wedge H_{k^{\prime\prime}}^{\ell^{\prime\prime}}))={\EuScript C}_k^{\ell}\wedge ({\EuScript C}_{k^{\prime}}^{\ell^{\prime}}\vee{\EuScript C}_{k^{\prime\prime}}^{\ell^{\prime\prime}})$\\
\item[6)] $\rho(\Omega_{H_k^{\ell}})=\Sigma_{{\EuScript C}_k^{\ell}}$
\end{description}
\end{lem}
\begin{proof}
The proof of item (1) follows from definition \ref{6tre453}\\
For (2) 
\begin{align*}
\rho(H_k^{\ell}\wedge  H_{k^{\prime}}^{\ell^{\prime}})&=\rho(H_{\min\{k, k^{\prime}\}}^{\min\{\ell, \ell^{\prime}\}})\\
                                                                                  &={\EuScript C}_{\min\{k, k^{\prime}\}}^{\min\{\ell, \ell^{\prime}\}}\\
                                                                                  &={\EuScript C}_k^{\ell} \vee {\EuScript C}_{k^{\prime}}^{\ell^{\prime}}\\
                                                                                  &=\rho(H_k^{\ell})\vee  \rho(H_{k^{\prime}}^{\ell^{\prime}})
\end{align*}
The proof of (3) is similar to that of item (2). The proofs of (4) and (5) follow from items (2) and (3).\\
For (5);   
Let $H_{\alpha}^{\beta}\in \Omega_{H_k^{\ell}}$ then by definition  $H_{\alpha}^{\beta}\leq H_k^{\ell}$ and ${\EuScript C}_k^{\ell} \leq {\EuScript C}_{\alpha}^{\beta}$
therefore  $\rho(\Omega_{H_k^{\ell}})\subseteq\Sigma_{{\EuScript C}_k^{\ell}}$. The other contention is similar.
\end{proof}
\begin{prop}\label{1som0rph5654}
The Lattices ${\mathscr FM}$ and ${\mathscr H}$ are isomorphis and $\Sigma_{{\EuScript C}_k^{\ell}}$ is a principal ideal.
\end{prop}
\begin{proof}
As $\rho$ is a bijection, then by  lemma \ref{Ttre4321}  $\rho$ is an isomorphism of lattice algebras.
By proposition \ref{id3al8766} it is concluded that $\Sigma_{{\EuScript C}_k^{\ell}}$ is an ideal.
\end{proof}
As in section \ref{231yytrfde}, we denote by $T_{{\EuScript C}_k^{\ell}}$ the Tanner graph  and  by $g({\EuScript C}_k^{\ell})$ the girth of ${\EuScript C}_k^{\ell}$, respectively.
\begin{prop}\label{hujyy5477}
If ${\EuScript C}_{k^{\prime}}^{\ell^{\prime}}\in \Sigma_{{\EuScript C}_k^{\ell}}$, then ${\EuScript C}_k^{\ell}$ is a subcode of ${\EuScript C}_{k^{\prime}}^{\ell^{\prime}}$,
$T_{{\EuScript C}_{k^{\prime}}^{\ell^{\prime}}} \subseteq T_{{\EuScript C}_k^{\ell}}$\\ 
and $g({\EuScript C}_k^{\ell})\leq g({\EuScript C}_{k^{\prime}}^{\ell^{\prime}})$
\end{prop}
\begin{proof}
If ${\EuScript C}_{k^{\prime}}^{\ell^{\prime}}\in \Sigma_{{\EuScript C}_k^{\ell}}$ then ${\EuScript C}_k^{\ell}\leq {\EuScript C}_{k^{\prime}}^{\ell^{\prime}}$ and $H_{k^{\prime}}^{\ell^{\prime}}\leq H_k^{\ell}$ which implies $H_{k^{\prime}}^{\ell^{\prime}}$ is submatrix $H_k^{\ell}$. 
Clearly  ${\EuScript C}_k^{\ell}=\ker H_k^{\ell}\subseteq \ker H_{k^{\prime}}^{\ell^{\prime}} ={\EuScript C}_{k^{\prime}}^{\ell^{\prime}}$ then ${\EuScript C}_k^{\ell}$ is a subcode of ${\EuScript C}_{k^{\prime}}^{\ell^{\prime}}$. For Tanner graph we have  $T_{{\EuScript C}_{k^{\prime}}^{\ell^{\prime}}} \subset T_{{\EuScript C}_k^{\ell}}$ and this implies 
$g({\EuScript C}_k^{\ell})\leq g({\EuScript C}_{k^{\prime}}^{\ell^{\prime}})$
\end{proof}
\begin{cor}
If ${\EuScript C}_k^{\ell}\in{\mathscr H}$ then the girth $g({\EuScript C}_k^{\ell})=6$ 
\end{cor}
\begin{proof}
If ${\EuScript C}_k^{\ell}\in{\mathscr H}$ then ${\EuScript C}_k^{\ell}=\ker H_k^{\ell}$. By  theorem \ref{gfbd3233} the girth $g({\EuScript C}_k^{\ell})=6$
\end{proof}
\begin{lem}
Let $k, \ell \geq 2$ and $0\leq r \leq \min\{k-1, \ell-1 \}$. Then \\
$${\EuScript C}_k^{\ell}=\bigwedge_{i=0}^r\binom{r}{i} {\EuScript C}_{k-i}^{\ell-(r-i)} $$
\end{lem}
\begin{proof}
By lemma \ref{dfrim43443} we have $H_k^{\ell}=\bigvee_{i=0}^r\binom{r}{i} H_{k-i}^{\ell-(r-i)}$.

Applying the function $\rho$ (see, \ref{treeew321} ) we have 
\begin{align*}
\rho (H_k^{\ell})&=\rho(\bigvee_{i=0}^r\binom{r}{i} H_{k-i}^{\ell-(r-i)})\\
                        &=\bigwedge_{i=0}^r\binom{r}{i} \rho(H_{k-i}^{\ell-(r-i)})\\ 
                        &=\bigwedge_{i=0}^r\binom{r}{i} {\EuScript C}_{k-i}^{\ell-(r-i)}
\end{align*}
\end{proof}
\begin{figure}[H]
\centering

\begin{tikzpicture}[>=stealth,scale=0.8]

\draw[->,thick] (0,0) -- (8,0) node[right] {$k$};
\draw[->,thick] (0,0) -- (0,5) node[above] {$\ell$};

\coordinate (A) at (5,4);
\coordinate (B) at (6,4);
\coordinate (C) at (6,3);
\coordinate (D) at (5,3);
\coordinate (D1) at (3,4);
\coordinate (E) at (6,2);
\coordinate (F) at (4,4);
\coordinate (X) at (0,0);
\coordinate (Y) at (1,1);
\coordinate (Z) at (2,2);
\coordinate (M) at (3,3);
\coordinate (M1) at (6,1);
\coordinate (M2) at (6,0);
\coordinate (N1) at (1,4);
\coordinate (N2) at (0,4);

\foreach \p in {A,B,C,D,D1,E,F,X,Y,Z,M,M1,M2,N1,N2}
\fill (\p) circle (1.5pt);

\draw[->,thick] (B) -- (A);
\draw[->,thick] (B) -- (C);
\draw[->,thick] (A) -- (D);
\draw[->,thick] (C) -- (D);
\draw[->,thick] (C) -- (E);
\draw[->,thick] (A) -- (F);
\draw[->,thick] (Y) -- (X);
\draw[->,thick] (Z) -- (Y);
\draw[->,thick] (M) -- (Z);
\draw[->,thick] (F) -- (D1);
\draw[->,thick] (M1) -- (M2);
\draw[->,thick] (N1) -- (N2);

\node[above] at (A) {${\EuScript C}_{k-1}^{\ell}$};
\node[right] at (B) {${\EuScript C}_k^{\ell}
={ \EuScript C}_k^{\ell-1} \vee {\EuScript C}_{k-1}^{\ell}$};
\node[right] at (C) {${\EuScript C}_k^{\ell-1}$};
\node[left] at (D) {${\EuScript C}_{k-1}^{\ell-1}$};
\node[above] at (Y) {${\EuScript C}_1^1$};
\node[above] at (Z) {${\EuScript C}_2^2$};
\node[right] at (M1) {${\EuScript C}_k^1$};
\node[below] at (M2) {${\EuScript C}_k^0$};
\node[above] at (N1) {${\EuScript C}_1^{\ell}$};
\node[left] at (N2) {${\EuScript C}_0^{\ell}$};
\node[above] at (F) {${\EuScript C}_{k-2}^{\ell}$};
\node[right] at (E) {${\EuScript C}_k^{\ell-2}$};

\end{tikzpicture}

\caption{The diagram sketches the ideal
$\Sigma_{{\EuScript C}_k^\ell}$. The arrows indicate inclusions and show the
three-decomposition
${\EuScript C}_k^\ell
= {\EuScript C}_k^{\ell-2}
\wedge {\EuScript C}_{k-1}^{\ell-1}
\wedge {\EuScript C}_{k-2}^\ell$.}
\label{fig}
\end{figure}
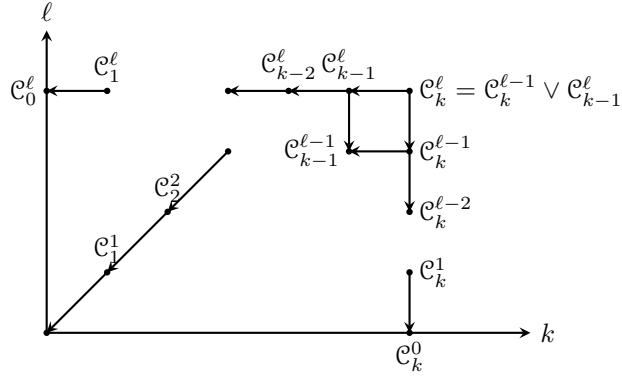
\begin{lem}\label{DFRE67yt}
Suppose that $\operatorname{char}{\mathbb F}\leq \ell$. Then
\begin{description}
\item[ 1)]  $\operatorname{rank}H_k^{\ell}<C^{k+\ell-1}_{\ell-1}$ if $k\geq\ell$
\item[ 2)]  $\operatorname{rank}H_k^{\ell}<C^{k+\ell-1}_{k-1}$ if $k\leq\ell$
\end{description}
\end{lem}
\begin{proof}
The proof proceeds by mathematical induction on $\ell$.\\
If $\ell=2$ then  $\operatorname{char}({\mathbb F})=2$ and
\begin{align*}
 H_k^2 &\sim \left[ \begin{tabular}{c|c}
$I_k$ & $ H_{k-1}^2 $ \cr\hline $0$  & $-2, \ldots, -2$
\end{tabular} \right] \\
&\sim \left[ \begin{tabular}{c|c}
$I_k$ & $ H_{k-1}^{\ell} $ 
\end{tabular} \right] 
\end{align*}
Therefore $\operatorname{rank}H_k^\ell<k+1.$\\
Our induction hypothesis states that,  $2\leq \ell' < \ell$.
If $\operatorname{char}{\mathbb F}\leq \ell^{\prime}$ then $\operatorname{rank}H_k^{\ell^{\prime}}<C^{k+\ell^{\prime}-1}_{\ell^{\prime}-1}.$\\
Consider the $(k,\ell)$-Fractus matrix $H_k^{\ell}$ of size $C_{\ell-1}^{k+\ell-1}\times C_{\ell}^{k+\ell-1}$, and suposse that 
$$H_k^{\ell}=\left[\begin{smallmatrix}
R_1\\
R_2 \\
\vdots\\
R_{C_{\ell-1}^{k+\ell-1}}
\end{smallmatrix}\right]$$
where $R_1,R_2,\ldots,R_{C_{\ell-1}^{k+\ell-1}}$ are the rows of $H_k^{\ell}$. Is easy see that $$(\ell, \ell,\ldots, \ell)=R_1+R_2+\cdots+R_{C_{\ell-1}^{k+\ell-1}}$$ and
$$H_k^{\ell}\sim\left[\begin{smallmatrix}
\ell, \ell,\ldots, \ell\\
R_2 \\
\vdots\\
R_{C_{\ell-1}^{k+\ell-1}}
\end{smallmatrix}\right].$$\\
Therefore, if \text{char}${\mathbb F}=\ell$, then \text{rank}$H_k^{\ell}<C_{\ell-1}^{k+\ell-1}$.\\  
If \text{char}${\mathbb F}<\ell$ then by \ref{m4tr1x67532} and induction hypothesis we have
\begin{align*}
\operatorname{rank}(H_k^\ell)
&= C_{\ell-1}^{k+\ell-2}+
\operatorname{rank}(H_k^{\ell-1}) \\
&\leq C_{\ell-1}^{k+\ell-2}+
C_{\ell-2}^{k+\ell-2}\\
&=C_{\ell-1}^{k+\ell-1}
\end{align*}
Therefore, if \text{char}${\mathbb F}\leq\ell$ and $\ell<k$, then \text{rank}$(H_k^{\ell})<C_{\ell-1}^{k+\ell-1}$.\\  
For item (2), by lemma \ref{FTR43eDD} and item (1) we have 
\begin{align*}
\operatorname{rank}(H_k^\ell)&=\operatorname{rank}(H_{\ell}^k)\\
                                                &<C_{k-1}^{k+\ell-1}
\end{align*}  
Therefore, if \text{char}${\mathbb F}\leq\ell$ and $\ell>k$, then \text{rank}$(H_k^{\ell})<C_{k-1}^{k+\ell-1}$.\\ 
\end{proof}
\begin{lem}\label{hgf56y7}
Suppose that \text{char}${\mathbb F}>\ell$. Then 
\begin{description}
\item[ 1)]  \text{rank}$H_k^{\ell}=C_{\ell-1}^{k+\ell-1}$ if $k\geq\ell$.
\item[ 2)]  \text{rank}$H_k^{\ell}=C_{k-1}^{k+\ell-1}$ if $k\leq\ell$.
\end{description}
\end{lem}
\begin{proof}
Let $p=\operatorname{char}\mathbb F>\ell$. We proceed by induction on $\ell$.
For $\ell=2$, we have $p>2$. In the block decomposition of $H_k^2$, the lower block contains the scalar $-2$, which is nonzero in $\mathbb F$. Therefore,
$$\operatorname{rank}H_k^2=
\left[ \begin{tabular}{c|c}
$I_k$ & $ H_{k-1}^2 $ \cr\hline $0$  & $-2, \ldots, -2$
\end{tabular} \right]
=C^{k+2-1}_{2-1}=k+1.$$
Now suppose that the result holds for $\ell-1$, and that $\operatorname{char}{\mathbb F}>\ell$. In particular,
$$\operatorname{char}{\mathbb F}>\ell-1,$$
and hence, by the induction hypothesis,
$$\operatorname{rank}H_k^{\ell-1}=\binom{k+\ell-2}{\ell-2}.$$
Using the same block decomposition as in \ref{m4tr1x67532}, we obtain
$$\operatorname{rank}(H_k^\ell)=\left[ \begin{tabular}{c|c}
 $H_k^{\ell-1}$ & $ 0 $ \cr\hline $I_{C_{\ell-1}^{k+\ell-2}}$  & $H_{k-1}^{\ell}$
\end{tabular} \right]=C_{\ell-1}^{k+\ell-2}+\operatorname{rank}(H_k^{\ell-1}).$$
Substituting the induction hypothesis, we get
$$\begin{aligned}
\operatorname{rank}(H_k^\ell)
&=C^{k+\ell-2}_{\ell-1}+
C^{k+\ell-2}_{\ell-2}\\
&=C^{k+\ell-1}_{\ell-1},
\end{aligned}$$
where the last equality follows from Pascal's identity. Consequently,
$$\operatorname{char}{\mathbb F}>\ell
\quad\Longrightarrow\quad
\operatorname{rank}H_k^\ell=C^{k+\ell-1}_{\ell-1}.$$
Since $H_k^\ell$ has exactly $C^{k+\ell-1}_{\ell-1}$ rows, this shows that it has full row rank.\\
For item (2), by lemma \ref{FTR43eDD} and item (1) we have 
$$\operatorname{rank}(H_k^\ell)=\operatorname{rank}(H_{\ell}^k)=C_{k-1}^{k+\ell-1}$$ 
Therefore, if \text{char}${\mathbb F}>\ell$ and $\ell>k$, then \text{rank}$(H_k^{\ell})=C_{k-1}^{k+\ell-1}$.\\ 
\end{proof}
\begin{cor}
Let $k,\ell\geq 2$, and let $r_{{\EuScript C}_k^{\ell}}$ denote the rate of the $(k,\ell)$-Fractus code ${\EuScript C}k^{\ell}$. Then:
\begin{enumerate}
\item Let $\operatorname{char}\mathbb F>\ell$,\\
if $k>\ell$ then $r_{{\EuScript C}_k^{\ell}}=1-\frac{\ell}{k}$  or if $k<\ell$ then $r_{{\EuScript C}_k^{\ell}}=1-\frac{k}{\ell}.$ 
\item Let $\operatorname{char}\mathbb F\leq\ell$,\\
if $k>\ell$ then $r_{{\EuScript C}_k^{\ell}}>1-\frac{\ell}{k}$ or  if $k<\ell$ then  $r_{{\EuScript C}_k^{\ell}}>1-\frac{k}{\ell}.$ 
\end{enumerate}
\end{cor}
\begin{proof}
For item (1) by lemma \ref{hgf56y7} if \text{char}${\mathbb F}>\ell$, then \text{rank}$H_k^{\ell}=C_{\ell-1}^{k+\ell-1}$, then 
$$r_{{\EuScript C}_k^{\ell}}=1-\frac{\operatorname{rank}H_k^\ell}{C_\ell^{k+\ell-1}}=1-\frac{\ell}{k}.$$
The evidence provided by the other party is similar.\\
On the other hand, if $\operatorname{char}\mathbb F\leq\ell$, then
$\operatorname{rank}H_k^\ell<C_{\ell-1}^{k+\ell-1}.$
Consequently,
$$1-\frac{\ell}{k}<1-\frac{\operatorname{rank}H_k^\ell}{C_\ell^{k+\ell-1}}=r_{{\EuScript C}_k^{\ell}}$$
The other party provides similar evidence.
\end{proof}
\section{Conclusions}
In this paper, we introduced Fractus matrices, a new family of recursively generated sparse matrices characterized by self-similarity, regularity, and hierarchical structure. These matrices provide a unified algebraic framework for constructing regular LDPC codes while preserving the sparsity required for efficient iterative decoding.
Several structural properties of Fractus matrices were established, including recursive decomposition, flip-transpose symmetry, sparsity, and regularity. These properties naturally lead to recursive algorithms for solving the associated linear systems and for encoding the corresponding LDPC codes. The resulting encoding procedure has nearly linear computational complexity and exploits the recursive nature of the proposed matrices without requiring the explicit computation of dense generator matrices.
From the coding-theoretic perspective, the proposed construction produces LDPC codes whose Tanner graphs are free of 4-cycles and have girth six. Consequently, the minimum Hamming distance satisfies the Tanner tree bound, while iterative decoding can be performed using standard message-passing algorithms with complexity proportional to the number of graph edges.
One of the main contributions of this work is the identification of a lattice structure associated with the family of Fractus matrices and the corresponding Fractus codes. This connection establishes a bridge between lattice theory and coding theory, providing an algebraic viewpoint that complements the traditional graph-theoretic approach to LDPC codes. In particular, the lattice framework offers a natural setting for studying recursive constructions, structural symmetries, ideals, and hierarchical relationships among parity-check matrices and their associated codes.
Future research includes the study of lattice automorphisms and ideals associated with Fractus matrices, the characterization of the resulting lattices, the analysis of trapping and stopping sets through the recursive structure, the extension of the proposed construction to irregular and non-binary LDPC codes, and the optimization of the recursive encoding and decoding algorithms for hardware implementations and large-scale communication systems.
\section*{Declaration on the Use of AI Technologies}

The author used ChatGPT (OpenAI)  as supportive tools during the preparation of this work. These tools were employed to assist with literature searches, verify computations, facilitate mathematical discussions that contributed to refining key arguments, and improve the clarity and style of the manuscript. All definitions, proofs, original research ideas, and the final content of the manuscript are the sole responsibility of the author. The author assumes full responsibility for the accuracy, integrity, and content of this publication.
%

\end{document}